\documentclass[a4paper,11pt,twoside]{article}
\usepackage[a4paper,left=2.2cm,right=2.2cm,top=2.3cm,bottom=2.6cm]{geometry}
\usepackage[english]{babel}
\usepackage{setspace}

\usepackage{mathrsfs,hyperref, algorithm, algorithmic, yfonts}
\usepackage{harvard}
\usepackage[]{amsmath}
\usepackage{amsthm}
\usepackage{fix-cm}
\usepackage[]{amssymb}
\usepackage[]{latexsym}
\usepackage[latin1]{inputenc}
\usepackage[right]{eurosym}
\usepackage[T1]{fontenc}
\usepackage[]{graphicx}
\usepackage[]{epsfig}
\usepackage{fancyhdr}
\usepackage{bbm}
\usepackage{pstricks}
\usepackage{multirow}
\makeatletter

\renewcommand{\p@enumi}{theenumi-}
\renewcommand{\@fnsymbol}[1]{\@alph{#1}}

\newcommand{\bbr}{\mathbb{R}}

\newcommand{\bby}{\mathbb{Y}}
\newcommand{\bbx}{\mathbb{X}}

\newcommand{\ci}{\citeasnoun}

\newcommand{\fil}{\mathcal{F}}
\newcommand{\fcal}{\mathcal{F}}

\newcommand{\pcal}{\mathcal{P}}

\newcommand{\bcal}{\mathcal{B}}

\newcommand{\kcal}{\mathcal{K}}
\newcommand{\acal}{\mathcal{A}}

\newcommand{\xcal}{\mathcal{X}}
\newcommand{\ycal}{\mathcal{Y}}

\newcommand{\ucal}{\mathcal{U}}

\newcommand{\el}{\ell}

\newcommand{\ga}{\alpha}

\DeclareMathOperator\sgn{sgn}

\newcounter{modcount}
\newcommand{\modulo}[2]{%
\setcounter{modcount}{#1}\relax
\ifnum\value{modcount}<#2\relax
\else\relax
\addtocounter{modcount}{-#2}\relax
\modulo{\value{modcount}}{#2}\relax
\fi}

\newcommand{\tablepictures}[4][c]{\begin{tabular}[#1]{@{}c@{}}#2\vspace{0.5cm}\\(\alph{#4}) #3\end{tabular}}
\newcounter{gridsearch}
\newcommand{\tabpic}[2]{
    \stepcounter{gridsearch}
    \modulo{\thegridsearch}{2}
    \ifnum\value{modcount}=0
        \tablepictures[t]{#1}{#2}{gridsearch}\\[2.0cm]
    \else
        \tablepictures[t]{#1}{#2}{gridsearch}&~&
    \fi
}

\makeatother
\newtheorem{lemma}{Lemma}[section]
\newtheorem{proposition}[lemma]{Proposition}

\newtheorem{defi}[lemma]{Definition}
\newtheorem{example1}[lemma]{Example}
\newtheorem{ex1}[lemma]{Example}
\newtheorem{rem1}[lemma]{Remark}
\newtheorem{assumption}[lemma]{Assumption}
\newtheorem{alg1}[lemma]{Algorithm}
\newtheorem{me1}[lemma]{Mechanism}

\newenvironment{rem}{\begin{rem1}\rm}{\end{rem1}}
\newenvironment{ex}{\begin{ex1}\rm}{\end{ex1}}
\newenvironment{example}{\begin{example1}\rm}{\end{example1}}

\usepackage{color}
\newcommand{\T}{\mathsf{T}}

\DeclareMathOperator{\Min}{Min}

\DeclareMathOperator*{\argmin}{arg\,min}

\begin{document}
%\pagestyle{fancy}
%\fancyfoot{}% Unten nichts
%\fancyhead[EL,OR]{{\tt\thepage}}
%\fancyhead[ER]{{\tt Knispel\,\&\,Weber}}
%\fancyhead[OL]{{\tt Titel}}
%\fancyhead[EC,OC]{}

\title{Measures of Systemic Risk}
\author{ Zachary Feinstein\footnote{Zachary Feinstein,  ESE, Washington University, St. Louis, MO 63130, USA, {\tt zfeinstein@ese.wustl.edu}}\\[0.7ex] \textit{Washington University in St. Louis} \and  Birgit Rudloff\footnote{Birgit Rudloff, Wirtschaftsuniversit\"at Wien,
Institute for Statistics and Mathematics,
Geb\"aude D4,
Welthandelsplatz 1, 1020 Wien, Austria, {\tt birgit.rudloff@wu.ac.at}}\\[0.7ex] \textit{Wirtschaftsuniversit\"at Wien}  \and  Stefan Weber\footnote{Stefan Weber, IfMS, Leibniz Universit\"at Hannover, 30167 Hannover, Germany, {\tt sweber@stochastik.uni-hannover.de}}\\[0.7ex] \textit{Leibniz Universit{\"a}t Hannover } }
\date{\today~(Original: February 27, 2015)\footnote{ The authors are grateful to the editors and referees for thoughtful comments and encouragements that led to a greatly improved paper.}}
\maketitle

\begin{abstract}
\emph{Systemic risk} refers to the risk that the financial system is susceptible to failures due to the characteristics of the system itself. The tremendous cost of systemic risk requires the design and implementation of tools for the efficient macroprudential regulation of financial institutions. The current paper proposes a novel approach to measuring systemic risk.
 
Key to our construction is a rigorous derivation of systemic risk measures from the structure of the underlying system and the objectives of a financial regulator.  The suggested \emph{systemic risk measures} express systemic risk in terms of capital endowments of the financial firms.  Their definition requires two ingredients: a \emph{cash flow or value model} that assigns to the capital allocations of the entities in the system a relevant stochastic outcome; and an \emph{acceptability criterion}, i.e.\ a set of random outcomes that are  acceptable to a regulatory authority. Systemic risk is  measured by \emph{the set of allocations of additional capital that lead to acceptable outcomes}. We explain the conceptual framework and the definition of systemic risk measures, provide an algorithm for their computation, and illustrate their application in numerical case studies. 
 
Many systemic risk measures in the literature can be viewed as the minimal amount of capital that is needed to make the system acceptable after aggregating individual risks, hence quantify the costs of a bail-out. In contrast, our approach emphasizes operational systemic risk measures that  include both \emph{ex post} bailout costs as well as \emph{ex ante} capital requirements and may be used to prevent systemic crises. 
 
\end{abstract}\vspace{0.2cm}

%\doublespace

\section{Introduction}

\emph{Systemic risk} refers to the risk that the financial system is susceptible to failures due to the characteristics of the system itself. The tremendous cost of this type of risk requires the design and implementation of tools for the efficient macroprudential regulation of financial institutions. The current paper proposes a novel and general approach to systemic risk measurement.  The suggested risk statistics summarizes complex features of financial systems in a simple manner and provides at the same time a unified theory that includes many previous contributions. Whenever a suitable model of a financial system is available that allows to mimic certain types of interventions, their impact on systemic risk can quantitatively and qualitatively be analyzed using our measurement tools. In numerical case studies, we also demonstrate that our systemic risk measure which is multi-variate captures properties that cannot be detected by previous univariate risk statistics.  

Key is a rigorous derivation of systemic risk measures from the structure of the underlying system and the objectives of a financial regulator. Our definition of  systemic risk measures requires two ingredients: first, a \emph{cash flow or value model} that assigns to capital allocations of the entities in the system a relevant stochastic outcome. For example, the capital allocation in a complex financial system might be mapped to the stochastic benefits and costs of the society over a fixed time horizon. Mathematically, the cash flow or value model can be described by a non-decreasing random field $Y=(Y_k)_{k\in\bbr^l}$ that provides for each capital allocation $k=(k_i)_{i=1, 2, \dots, l}$ a random variable $Y_k$ that captures the system output. The second ingredient that is required is an acceptability criterion $\acal$; this is the subset of random variables that are considered as acceptable outcomes from the point of view of a regulatory authority. Following the key ideas of monetary risk measures, systemic risk is then simply measured by \emph{the set of allocations of additional capital that lead to acceptable outcomes}; i.e., if $k\in \bbr^l$ is the current capital allocation, then the systemic risk measure $R(Y;k)$ is defined by:
$$R(Y;k) = \left\{ m \in \bbr^l \;  | \;   Y_{k+m} \in \acal \right\}. $$
This is a \emph{set-valued and multivariate risk measure}. We develop an algorithm for the explicit computation of these risk measures and also show how they could be used to derive capital requirements for regulatory purposes in a system of banks. The risk measure is cash-invariant and by characterizing the additional capital that achieves acceptability it provides an \emph{operational perspective} on systemic risk.  In a series of case studies we highlight the features of our systemic risk measure in the context of systemic risk aggregation, see e.g.\  \ci*{chen2013axiomatic}, and  network models as suggested in the seminal paper of \ci{EN01}, see also \ci{CFS05}, \ci{RV11},  \ci{AW_15}, \ci{fei15}, and \ci{hurd15}.

Most systemic risk measures proposed in the literature are of the form
\begin{equation}\label{ins}
\rho(\Lambda(X))=\inf\{k\in\bbr \; | \; \Lambda(X)+k\in \acal\};
\end{equation}
here, an aggregation function $\Lambda:\bbr^n\rightarrow \bbr$ is first applied to an $n$-dimensional risk factor $X$ in order to obtain a univariate random variable $\Lambda(X)$; second, a classical scalar risk measure $\rho$, corresponding to the acceptance set $\acal$, is evaluated at the univariate random variable $\Lambda(X)$.\footnote{  
\ci*{chen2013axiomatic} provide an axiomatic definition of measures of systemic risk with  a decomposition as a concatenation of a scalar risk measure and an aggregation functional  as in \eqref{ins}. Their axiomatic analysis is technically extended in \ci*{kromer2013systemic}, see also \ci*{MB2015}, who work on general probability spaces and consider risk measures that are not necessarily positively homogeneous.
Further examples of the form \eqref{ins} are the Contagion Index of \ci{CMS10} or SystRisk of \ci{brunnermeier2013measuring}.} 
From an economic point of view, the aggregation function $\Lambda$ can be interpreted as the system mapping that captures how the individual components of the system interact and affect other entities, e.g.\ the \emph{society}. Systemic risk measures of the form \eqref{ins}  quantify the financial resources that are needed to move the system into an acceptable state \emph{after} the systemic impact has been observed -- possibly after a crisis has already occurred. They thus represent the costs of a bail-out. 

Capital regulation in contrast provides an instrument that aims at \emph{preventing} crises. This can be captured by systemic risk measures of the form
\begin{equation}\label{sens}
\inf\left\{\sum_{i=1}^n k_i \; | \; \Lambda(X+k)\in \acal\right\},
\end{equation}
for $k=(k_1,...,k_n)\in\bbr^n$. Allocations of additional capital are in this case added to $X$ \emph{before} aggregating risk  components -- not \emph{afterwards} as in \eqref{ins}. Equation \eqref{sens} can typically not be transformed into equation \eqref{ins}, unless $\Lambda$ is trivial. 

This paper provides a general framework that includes both bail-out risk measures \eqref{ins} as well as initial capital requirements \eqref{sens}. In addition, it turns out to be unnecessary to restrict attention to finitely many driving factor and a single aggregation function, i.e.\ $(X, \Lambda)$ as in equation \eqref{ins} and \eqref{sens}. A general cash flow 
or value model is appropriately described by increasing random fields. Moreover, our approach includes other systemic risk measures found in the literature such as CoV@R of \ci{adrian2011covar}.

The contributions of this paper are as follows:
 \begin{itemize}
 \item The suggested multivariate approach to systemic risk is novel. It constructs systemic risk measures from the essential ingredients of systemic risk: the available capital input $k$, a cash flow or value model of the relevant stochastic outcomes -- captured by a random field  $Y$, and an acceptability criterion $\acal$. 
 \item We provide a framework into which -- in many cases -- the most important previous contributions of the literature, \ci*{chen2013axiomatic}, \ci*{brunnermeier2013measuring}, and  \ci{adrian2011covar}, can be embedded. Our construction is, however, also more general than the literature in the sense that it allows for very complex feedback mechanisms, as we will illustrate in Section~\ref{Sec:casestudies}. The risk measures could also be applied in sophisticated simulation models, see e.g.\  \ci{Borovkova13}.
 \item  The suggested systemic risk measures, despite the fact that they are set-valued, can be computed explicitly. We construct a simple algorithm and demonstrate in numerical case studies that this procedure can easily be implemented. 
 \item In Section~\ref{Sec:Orthant} we introduce a new concept called \emph{efficient cash-invariant allocation rules (EARs)}  which allows to derive specific capital allocations from systemic risk measures that could be used as capital requirements within a banking system. While the suggested systemic risk measures are set-valued, EARs are usually single-valued. They can be characterized as minimizers of the weighted cost of capital on a global level; for this purpose, prices of regulatory capital  are assigned to all entities in the financial system. They provide a general framework for many attribution rules that are suggested in the literature.
 \item We also show in our numerical case studies that systemic risk measures provide significantly more information about systemic risk than previous approaches or EARs. This highlights fundamental challenges related to univariate systemic risk measures and attribution rules. 
  \end{itemize}

\paragraph{Literature.} 
There is an extensive literature on the axiomatic theory of monetary risk measures that was initiated by \ci{ADEH99}, see e.g.\ \ci{foellmer-schied3rd} and \ci{FW15} for surveys. Set-valued risk measures are e.g.\ considered in  \ci{jt04}, \ci{hamel2010duality},  \ci{hamel2011set},  and \ci{cascos}. 

In the current paper, the cash flow or value model that captures the relevant stochastic outcome is described by a non-decreasing random field; this allows, in particular, a non-linear dependence of the final position on capital. Such a situation is typical in the presence of systemic risk or illiquidity, see e.g.\  \ci{Liquidity}. Random fields substantially generalize previous approaches. 

Our systematic approach to systemic risk is new to the literature. At the same time,  previous constructions can -- in many situations --  be interpreted as special cases of our systemic risk measures: If the systemic risk measure of  \ci*{chen2013axiomatic} is coherent, or if the risk measure is formulated within a conceptual framework as described in \ci*{ADKM09}, it can be embedded into our setting.  Furthermore, the construction of \ci{brunnermeier2013measuring} can be interpreted as a special case of our approach. Also other concepts including  conditional systemic risk measures like  CoV@R of \ci{adrian2011covar} fit into our  methodology.  Furthermore, on the basis of our framework many examples from the literature can be modified in such a way that they can be used as capital requirements that incorporate their impact on the system.

Another strand of literature considers market-based indices of systemic distress that measure the expected capital shortfall of a bank conditional on systemic events, as suggested in \ci{Acharya} and \ci{SRISK}. Unlike  \ci{adrian2011covar} these papers focus on  conditional expectations instead of conditional V@Rs that they estimate from data. Their results can e.g.\ be used in order to rank systemically risky firms.

Independently from our paper, a closely related systemic risk measure was developed by \ci{BFFM2015}. Their working paper was available about one month after we published our working paper on arXiv: \href{http://arxiv.org/pdf/1502.07961v1.pdf}{http://arxiv.org/pdf/1502.07961v1.pdf}. In contrast to us, \ci*{BFFM2015} focus on univariate systemic risk measures and study risk measurement and attribution separately.  They do also allow for a random allocation of risk to individual entities.  If the risk allocation is deterministic, their approach can be embedded into ours as a special case.  A similar methodology is also applied in \ci{ACDP2015}. The follow-up paper \ci{Ararat2016} investigates robust representations of systemic risk measures.

\paragraph{Outline.} The paper is structured as follows. In Section \ref{Sec:SysRisk} we explain our conceptual framework and the definition of the suggested systemic risk measures. Section \ref{Sec:Orthant} introduces  efficient cash-invariant allocation rules.  Section \ref{Sec:grid} develops algorithms for the computation of systemic risk measures. We illustrate the implementation of the procedure in numerical case studies, first, in the setting of \ci*{chen2013axiomatic} and \ci*{kromer2013systemic} and, second, in the network models of \ci{EN01} and \ci*{CFS05}. Section \ref{Sec:casestudies} reviews the structure of these models and describes the numerical results. The proofs are collected in the appendix.

\section{Measures of systemic risk}\label{Sec:SysRisk}

{ Let $(\Omega,\fcal,P)$ be a probability space. For families of random vectors and variables on  $(\Omega,\fcal,P)$, we use the following notation: $L^0(\Omega; \bbr^d)$ denotes the family of $d$-dimensional random vectors, $L^p(\Omega; \bbr^d) \subseteq L^0(\Omega; \bbr^d)$  and $L^\infty(\Omega; \bbr^d) \subseteq L^0(\Omega; \bbr^d)$  the subspaces of $p$-integrable and bounded random vectors, respectively. Additionally, we fix a topological vector space $\xcal \subseteq L^0(\Omega; \bbr)$ that contains the constants, for example   $\xcal = L^0(\Omega; \bbr)$, or $\xcal = L^{\infty}(\Omega; \bbr)$. }

We consider a single period model for the evolution of the economy. A definition of systemic risk requires two ingredients: the specification of a \emph{non-decreasing random field} and a notion of \emph{acceptability}. 
The index set of the random field will be interpreted as the amount of the eligible assets, or the \emph{capital allocation}. The random field  provides a random cash flow or value model that assigns a random output to each capital allocation. This output is either acceptable or not. Risk is then measured as the collection of capital allocations that need to be added in order to obtain acceptable outputs. 

\subsection{Eligible assets, and a cash flow or value model}\label{sec:cfvm}

We denote by 
$\ycal$
a space of non-decreasing random fields $Y: \bbr^l \to \xcal$; a random field $Y$ is non-decreasing if $k\leq m$, $k,m \in \bbr^l$, implies $Y_k \leq Y_m$. The random field provides a description of a cash flow or value model  (CVM).  Suppose that a collection of $l$ entities, e.g., financial firms or groups of firms is endowed with an amount of capital $k\in \bbr^l$, then $Y_k$ signifies a random system output at capital level $k$ such as the resulting total value of the financial system to the real economy. Monotonicity means that a larger capital allocation leads to a higher random value of the system to the economy. 

We are now going to discuss a few concrete examples of  CVMs  that can be used to capture systemic risk. 

\begin{example}\label{network_ex} Consider an interconnected economy of financial institutions $$N = \{1, 2, \dots, n\}.$$ These will typically be interpreted as banks, but could also include pension funds, hedge funds, and other financial entities. 
\begin{enumerate}
\item {\bf Aggregation mechanism insensitive to capital levels:} The setting of \ci*{chen2013axiomatic} is a special case of the proposed framework.  Letting $X \in L^0(\Omega; \bbr^n)$ be the future wealths of the agents in the financial sector, total wealth is aggregated across firms in the system by some increasing function $$\Lambda: \bbr^n \to \bbr .$$  In addition, each firm $i$ might be required to provide a capital amount $k_i$ to the system. In this case $l=n$ and the final random output  is then given by  $$Y_k := \Lambda(X) + \sum_{i = 1}^n k_i, \quad k \in \bbr^n.$$  
Observe that this approach aggregates the random endowments of the firms prior to considering capital endowments. These are simply added to the aggregated amount.

\item {\bf Aggregation mechanism sensitive to capital levels:} The setting of \ci*{chen2013axiomatic} described in Example (i) can be modified  by setting $$ Y_k := \Lambda(X+k) , \quad    k \in \bbr^n.$$ In this case, the capital allocation changes the argument of the aggregation function. This implies that the aggregation mechanism is sensitive to capital levels. In particular, feedback from capital levels to the final outputs are captured. 

\item {\bf A financial network with market clearing:} Assume the wealth of each agent randomly evolves between time $t=0$ and $t=1$ due to their primary business activities.  This might be captured, as e.g.\ in \ci*{CFS05}, by two random vectors $X, S \in L^0(\Omega; \bbr^n)$.  We assume that these vectors are the only sources of randomness in the economy.  The components of the vector $X$ are interpreted as the random amount of cash at time $t=1$ that each of the $n$ firms owns; the components of $S$  signify the number of shares that the firms hold in an illiquid asset. Analogous to the seminal paper \ci{EN01}, \ci*{CFS05} assume that the firms have bilateral obligations within the financial system. When the obligations are cleared after time $1$, some holdings of the illiquid asset must be liquidated; the resulting decrease in price has an adverse impact on the wealth of the financial institutions in the system. 

We assume that the financial system does not exist in isolation, but is connected to a further entity, the \emph{society}. Formally, this is achieved by simply augmenting the financial network by another node -- $0$ -- representing  the society.  Suppose that the capital endowments of the financial firms are $(k_i)_{i=1,2,\dots,n}$. The network clearing mechanism of \ci*{CFS05} provides the value of the equity of the entities or, if entities default, the amount of losses they cause to the system. These quantities can be captured for each node by a deterministic function
$$e_i: \bbr^n \times \bbr^n  \to \bbr, \quad i=0,1,2, \dots, n .  $$

If one focuses on the wealth of society, a suitable   CVM   is provided by the random field
$$Y_k := e_0(X+k;S), \quad k \in \bbr^n.$$
This can be seen as a special case of example (ii) with aggregation function  $$\Lambda(\cdot) := e_0(\cdot;S) .  $$
This example can easily be adopted to other variations of the model of  \ci{EN01}, including -- besides fire sales (cf.\ \ci{CFS05}, \ci{fei15}, \ci{hurd15}) -- e.g.\  bankruptcy costs (cf.\ \ci{RV11}) or cross-holdings (cf.\ \ci{E07}); see also \ci{AW_15} for a comprehensive extension. 

\item {\bf Groups:} While the previous example assumes that the capital of the financial institutions can be chosen without any restriction, constraints on groups of firms can easily be implemented. Suppose, for example, that the firms consist of three groups; in each group all firms hold the same amount of capital, i.e.\ $$ \left( 
k_1, k_1, \dots, k_1,k_2, k_2, \dots, k_2,k_3, k_3, \dots, k_3
\right)^\T .$$
This can be captured by a monotonously increasing function 
$g: \bbr^l \to \bbr^n, $
setting 
$$Y_k=e_0(g(k)+X;S),$$
where $l=3\leq n$ and $g(k) = ( 
k_1, \dots, k_1,k_2, \dots, k_2,k_3, \dots, k_3
)^\T$ in the specific example. This formalism can easily be adopted to other monotone constraints or the setting in Examples (i) and (ii).
\item {\bf Conditional distributions:} 
\ci{adrian2011covar} propose  a notion called CoV@R which is based on conditional quantiles (or, more generally, on conditional risk measures) in order to assess the contributions of market participants to overall systemic risk. Letting the market consist of $n$ financial firms, the basic idea of CoV@R is to describe the financial system by a random variable $X^{\rm system}$ and the individual firms by random variables $X_1, \ldots, X_n$. The random variables represent some economically meaningful quantity such as equity.  Given a confidence level $q$, the ${\rm CoV@R}_q^{j | C(X_i) = c}$ of entity $j\in\{1,\dots, n, {\rm system}\}$ is the V@R at level $q$ of $X_j$ conditional on institution $i$ being in a given financial condition represented by $C(X_i)=c \in \bbr$ for some measurable function $C: \bbr\to \bbr$:
$$
P\left(X_j  \leq {\rm CoV@R}_q^{j| C(X_i)=c } \big| C(X_i) =c\right)=q.
$$ 
In this setting, a corresponding   CVM   $$Y ^{j | C(X_i)=c}$$ may be chosen by requiring that for $k\in \bbr$ (i.e., $l=1$), the distribution of 
$ Y_k ^{j | C(X_i)=c}$ is equal to the conditional distribution of $k + X_j$ given $C(X_i)=c$. Of course, an analogous construction is possible when $C$ does not only depend on the state of entity $i$, but of multiple entities.  Moreover,  CVMs like  $Y ^{j | C(X_i)=c}$  can also provide the basis for other conditional risk measures beyond the V@R.  
\item  {\bf Increasing processes:}

Finally, we would like to consider  a simple example where $Y=(Y_k)$ cannot be written in terms of  deterministic transformations of $Y_0$.  As before, we assume that $X\in L^0(\Omega; \bbr^n)$ is the random future  wealth of the agents of the financial system that is aggregated across firms  in the system by some increasing function $\Lambda: \bbr^n \to \bbr$. However, we modify the mechanism how additional capital affects the wealths of the financial entities.

Generalizing (ii), assume that capital $k\in \bbr^n$ does not necessarily represent a cash allocation, but refers to investments with random payoffs. To this end, we consider a kernel $\mu$ from $(\Omega, \fil)$ to $(\bbr_+, \bcal)$ where $\bcal$ denotes the Borel-sets. Any investment of capital is represented by a set $B\in \bcal$ and leads to the random payoff $\mu (\cdot, B)$. 

We may now specify how the wealths of the banks depends on capital. The capital investment of bank $i$ is modeled by a mapping  $\mathfrak{B}^i: \bbr \to \bcal$ that is decreasing on $(-\infty, 0]$ and increasing on $[0, \infty)$. The random wealth of bank $i$ with capital $k_i$ at time $1$ is then given by:
$$Z^i_{k_i} = X_i +  \sgn (k_i) \cdot \mu(\cdot, \mathfrak{B}^i(k_i)) .    $$
Setting $Z_k = (Z^1_{k_1}, Z^2_{k_2} , \dots , Z^n_{k_n})^T$, the CVM is defined by
$$Y_k = \Lambda (Z_k) . $$

\end{enumerate}
\end{example}

\subsection{Systemic risk measures}

In the current section, we define systemic risk measures for CVMs. To begin with, we recall an important elementary result from the theory of monetary risk measures, see e.g.\ \ci{foellmer-schied3rd}: \emph{Any scalar monetary risk measure can  be represented as a capital requirement.} This statement describes the following correspondence: For any risk measure, the family of positions with non-positive risk is the acceptance set of the risk measure. Elements of the acceptance set are called acceptable. The risk measurement of any position can be recovered as a capital requirement, i.e.\ the smallest monetary amount that needs to be added to the position to make it acceptable. Due to this simple fact, the axiomatic theory of risk measures could also be based on acceptance sets instead of risk measures. This observation provides a useful starting point for our construction of systemic risk measures.

{ To be precise,} let  $\acal \subseteq \xcal$ be the acceptance set of a static scalar monetary risk measure with the following defining properties: 
\begin{enumerate}
\item $\inf\{m \in \bbr\; | \; m \in \acal \} > -\infty$,
\item $M_1 \in \acal$, $M_2 \in \xcal$, $M_2 \geq M_1$ $\Rightarrow$ $M_2\in \acal$,
\item $\acal$ is closed in $\xcal$.
\end{enumerate}
Property (i) describes that not any deterministic monetary amount is acceptable, excluding trivial cases. Property (ii) states that positions that dominate acceptable positions are again acceptable. Property (iii) is a weak condition that holds, if a risk measure is lower semicontinuous which is in most cases either required or automatically satisfied, see e.g.\  \ci{foellmer-schied3rd}, \ci{kainarueschendorf}, and \ci{chli09}.  

The scalar risk measure corresponding to $\acal$ is the corresponding capital requirement
$$\rho(M) = \inf \{ m \in \bbr\; | \; m + M \in \acal\} , \quad M\in \xcal,$$
i.e.\ the smallest monetary amount that needs to be added to a position such that it becomes acceptable. 

For the definition of systemic risk measures, suppose now that the acceptance set $\acal$ is the acceptance criterion of the regulator, i.e.\ a random outcome is acceptable, if it lies in $\acal$. Given a 
 CVM  $Y\in \ycal$, an amount of capital $k\in \bbr^l$ will thus be considered to be acceptable, if 
$$ Y_k \in \acal .$$
We will identify systemic risk with the collection { of the allocations of additional capital that lead to acceptability}. 
\begin{defi}
\label{def:risk}
Let $\pcal(\bbr^l;\bbr^l_+) := \{B \subseteq \bbr^l\; | \; B = B + \bbr^l_+\}$ be the collection of upper sets with ordering cone $\bbr^l_+$. We call the function
 $$R: \ycal\times\bbr^l \to \pcal(\bbr^l;\bbr^l_+) $$  a \emph{systemic risk measure}, if for some acceptance set $\acal \subseteq \xcal$ of a scalar monetary risk measure:
\begin{align}
\label{eq:def-risk}
R(Y;k) &= \{m\in \bbr^l\; | \; Y_{k+m} \in \acal\}.
\end{align}
\end{defi}
\begin{rem}
The risk measure $R$ is a set-valued function that assigns to the   CVM   $Y$ and a current level of capital $k$ all those capital vectors that make the output acceptable.  The monotonicity of the   CVM   $Y$ and the properties of the acceptance set $\acal$ imply that with $m\in  R(Y;k)$ also the   set $m+\bbr^l_+ \subseteq R(Y;k)$,  i.e., $R$ maps into $\pcal(\bbr^l;\bbr^l_+) $.
\end{rem}

Systemic risk measures possess standard properties like cash-invariance and monotonicity that are known from scalar monetary risk measures. 
\begin{lemma}\label{Lem:ci-m} Let $R$ be a systemic risk measure as introduced in Definition~\ref{def:risk}. Then $R$ has the following properties:
\begin{enumerate}
\item {\bf Cash-invariance}:
$$ R(Y;k)+m = R(Y;k-m)  \quad(\forall k,m \in \bbr^l, \;\forall Y \in \ycal).$$
\item {\bf Monotonicity}: 
 $$ (\forall Y,Z \in \ycal\mbox{ s.t. }\forall k \in \bbr^l: \;Y_k \geq Z_k) \; \Rightarrow  \; (\forall k \in \bbr^l: \; R(Y;k) \supseteq R(Z;k)).$$
\item {\bf Closed values}:  Suppose that $\bbr^l \to \xcal, k \mapsto Y_k$ is continuous. Then $R(Y;k) $ is a closed subset of $\bbr^l$.
\end{enumerate}
\end{lemma}

Property (i) refers to the fact that if current capital is reduced by $m$, the required additional capital must be increased by $m$. Property (ii) states that if $Y$ is always larger than $Z$, any additional capital vector that makes the output acceptable for the   CVM   $Z$ and a current capital level $k$ will also do so for the ``better''   CVM   $Y$. Property (iii) is a regularity assumption that will be used later.

Due to its cash-invariance, it is in many cases sufficient to analyze  $R(Y;0)$, since $R(Y;k) + k = R(Y;0)$. In order to simplify the notation, we will frequently write
$R(Y)$ instead of $R(Y;0)$.

\begin{rem}\label{Rem:lawinvariant} A systemic risk measure can also be distribution-based, also called law-invariant, see e.g.\ \ci{FW15}.
Assume that $\acal$ is the acceptance set of a distribution-based scalar risk measure. Let $Y, Z \in\ycal$.  If for all fixed $k \in \bbr^l$, $Y_k$ and $Z_k$ have the same distribution, then { $R(Y) = R(Z).$}
\end{rem}

We will now investigate the issue of convexity. In the classical setting, this is related to the question of whether or not risk measures quantify diversification appropriately. Before answering this question, we consider another elementary property, namely the convexity of the risk measurements  $R(Y)$.  

\begin{lemma}\label{Lem:conv-set}
Suppose that $\acal$ is convex and that $\bbr^l \to \xcal, k \mapsto Y_k$ is concave.  Then $R(Y)$ is a convex subset of $\bbr^l$. 
\end{lemma}

Let us now turn to the issue of diversification.  For this purpose, it will be convenient to assume that the family $\ycal$  is enumerated by some index set $I$, i.e.\ $\ycal = (Y ^i)_{i\in I}$. We will not assume that $i\mapsto Y^i$ is injective, so $Y ^i = Y ^j$ for $i\neq j$ is allowed.

For $i, j \in I$ and $\alpha\in [0,1]$ we set
$$C_\alpha(i, j)_k:= \alpha Y^i_k+(1-\alpha )Y^j_k, \quad k\in \bbr^l .$$
This is  the diversification of the final outputs of the model.  Financial agents divide their exposure between $Y^i $ and $Y ^j$ and receive, given their capital allocation $k$, the convex combination of the original outputs that are governed by $Y^i$ and $Y^j$. From now on we will assume that $i, j\in I$ implies $C_\alpha (i,j) \in \ycal$ for any $\alpha \in [0,1]$. 

In the case of systemic risk, diversification might, however,  also concern the inputs of the model.  Consider the situation of Example~\ref{network_ex}(ii), setting 
$ Y^X_k := \Lambda (X+k), \; k\in \bbr^n, ${~for~}$ X\in  I := L^\infty(\Omega; \bbr^n) .  $
Suppose that diversification refers to a convex combination of the underlying risk factors $X$ and $X'$, i.e.
$$D_\alpha(X,X' )_k = \Lambda (k + \alpha X + (1-\alpha) X'), \quad k\in \bbr^n, \;\alpha\in [0,1] .$$
Unless $\Lambda$ is linear, we typically have $D_\alpha(X,X')_k \neq C_\alpha(X, X')_k$.   This is the reason why  we call any  mapping
$$ D : \left\{
\begin{array}{ccl}
I \times I \times  [0,1] & \to & \ycal  \\
(i, j, \alpha) & \mapsto  & D_\alpha(i,j)
\end{array}
\right.$$
a \emph{diversification operator}. 

\begin{lemma}\label{Propo:rm} Let 
$R: \ycal \times \bbr^l \to \pcal(\bbr^l;\bbr^l_+)  $
be a systemic risk measure as introduced in Definition~\ref{def:risk} with $\ycal = (Y^i)_{i\in I}$, and fix $\alpha\in [0,1]$ and $i, j \in I$. 
If $\acal$ is convex and 
\begin{equation}\label{mixture}
	 \quad D_\alpha(i, j)_k\geq C_\alpha(i, j)_k  \quad (\forall k \in \bbr^l),
\end{equation}
then 
\begin{equation}\label{eq:quasi-c}R(D_\alpha(i, j))  \supseteq R(Y^i)  \cap R(Y^ j ) . \end{equation}
\end{lemma}

If property~\eqref{eq:quasi-c} holds for any $\alpha\in [0,1]$ and at any $(i, j) \in I^2$, we call $R$  \emph{$D$-quasi-convex}. $D$-quasi-convexity has the interpretation that the diversification mechanism $D$ does not increase capital requirements:
If $k \in \bbr^l$ is an acceptable capital allocation for the two  CVMs  $Y^i$ and $Y^j$, i.e.\ $Y^i_{k} \in \acal$ and $Y^j_{k} \in \acal$, then the diversified system response $D_\alpha(i,j)$ would also be acceptable for capital level $k \in \bbr^l$ -- i.e.\ $D_\alpha(i,j)_{k} \in \acal$. Equation \eqref{mixture} provides a simple criterion, namely that the diversification mechanism $D$ dominates diversification $C$ of outputs.

\begin{example}{~} \label{ex:acceptance}
We will now briefly discuss a few examples. More details, including illustrative numerical case studies, can be found in Section~\ref{Sec:casestudies}.

The definition of a systemic risk measure relies on a  CVM,  see Section \ref{sec:cfvm} for specific examples. The second ingredient is an acceptance criterion formalized by an acceptance set $\acal$. Let us first recall some standard examples, see  \ci{foellmer-schied3rd}, \ci{BT07}, and \ci{FW15}, before we discuss examples of systemic risk measures.  For simplicity, we fix $\xcal = L^\infty (\Omega; \bbr )$, although all risk measures could be defined on larger spaces. 

\begin{enumerate}
\item {\bf Acceptance criteria:}
\begin{enumerate}
\item \emph{Average value at risk}: Given $\lambda \in (0,1)$, the \emph{value at risk} at level $\lambda \in (0,1)$ is defined by 
$ {\rm V@R}_\lambda (M) := \inf \{m \in \bbr\; | \; P[ M + m < 0] \leq \lambda \} ,$ $  M\in \xcal.$ The \emph{average value at risk} 
is the average of the value at risk below the level $\lambda$:
$${\rm AV@R}_\lambda(M) :=  \frac 1 \lambda \int _0 ^\lambda {\rm V@R}_\gamma (M) d \gamma, $$
and constitutes a coherent risk measure. Its acceptance set is
$$\acal = \{M \; | \; {\rm AV@R}_\lambda(M) \leq 0  \} = \{ M \; | \; \exists r \in \bbr : E[(r-M )^+] \leq r \lambda \},  $$ see Lemma 4.51 in \ci{foellmer-schied3rd}.
\item \emph{Utility-based shortfall risk}: Letting $\el: \bbr \to \bbr$ be an increasing, non-constant, convex loss function and $z$ be a threshold level in the interior of the domain of $\el$, we define the convex risk measure
$${\rm UBSR}(M) : = \inf \{  m\in\bbr \; | \;   E[ \el ( - M -m )] \leq z   \}.$$
The \emph{utility-based shortfall risk} (UBSR) of a position $M$ equals the smallest monetary amount $m$ such that the disutility $E[ \el ( - M -m )]$ is at most $z$. It can be characterized as the unique root $m^*$ of the equation
$E[ \el ( - M -m^* )] - z = 0. $
The acceptance set of UBSR can be represented by
$$\acal = \{ M \; | \;   E[\el (-M)] \leq z   \}. $$
\item \emph{Optimized certainty equivalent}:
Let $u: \bbr \to [- \infty , \infty)$ be a concave and non-decreasing utility function satisfying $u(0) = 0 $, and $u(x) < x \; \forall x$. The \emph{optimized certainty equivalent} (OCE) is defined by the map ${\rm OCE}: \xcal \to \bbr$ with
$$ {\rm OCE} (M) = \sup_{\eta \in \bbr}  \{   \eta + E[u(M- \eta)] \} . $$
The negative of an OCE defines a convex risk measure
$\rho (M) = - {\rm OCE} (M) . $
Its acceptance set can be represented by
$$ \acal = \{  M \; | \; \exists \eta \in \bbr :    \eta + E[u(M- \eta)] \geq 0   \} .$$
The risk measure ${\rm AV@R} _\lambda$, $\lambda \in (0,1)$, can be obtained as a special case for $ u(t) = \frac 1 \lambda \cdot t$, $t\leq 0$, and $u(t) =0 $, $t>0$.
\end{enumerate}
\item {\bf Systemic risk measures:} We  fix the acceptance set $\acal \subseteq \xcal$ of any scalar risk measure on $\xcal$, e.g.\ any of the acceptance sets described above. 

\begin{enumerate}
\item Let $\Lambda: \bbr^n \to \bbr$ be an aggregation function as described in Example~\ref{network_ex}. We assume that $\Lambda$ is continuous and set  $I= \xcal.$ 

We consider now the case of an aggregation mechanism sensitive to capital levels: 
$$ Y^X_k =  \Lambda( k + X),\quad k \in \bbr^n , \; X\in I.  $$
The function $k\mapsto Y^X_k$ is concave for any $X$, if $\Lambda$ is concave, and in this case the systemic risk measure $R$ has convex values. 

The diversification operator could naturally be defined by 
$$D_\alpha(X,X')_k  = \Lambda (k + \alpha X + (1-\alpha) X'), \quad k \in \bbr^n, \; \alpha \in [0,1], \; X, X' \in I . $$
Suppose that $\Lambda$ is concave, then:
\begin{eqnarray*}
D_\alpha(X,X')_k & = & \Lambda (\alpha (k + X) + (1-\alpha) (k+X')) \\
& \geq &\alpha \Lambda  (k + X)+ (1-\alpha) \Lambda(k+X')\\
&= & C_\alpha(X,X')_k.
\end{eqnarray*}
Lemma \ref{Propo:rm} implies that $R$ is $D$-quasi-convex  if $\acal$ is convex. 

If the aggregation mechanism is insensitive to capital levels, one obtains an analogous result by setting  
\begin{eqnarray*} 
Y^X_k &  =   & \sum_{i = 1}^n k_i + \Lambda(X),\quad\quad\quad\quad\quad\quad \quad \quad k \in \bbr^n , \; X\in I, \\
D_\alpha(X,X')_k &  = &  \sum_{i = 1}^n k_i + \Lambda ( \alpha X + (1-\alpha) X'), \quad\quad k \in \bbr^n, \; \alpha \in [0,1], \; X, X' \in I . 
\end{eqnarray*}

\item  As described in  Example~\ref{network_ex}, a financial network  with market clearing can be interpreted as a special case of the aggregation mechanism that is sensitive to capital levels. The analysis above thus also applies in this case, implying that the concavity of $e_0$ in its first argument  and the convexity of $\acal$  leads to a $D$-quasi-convex systemic risk measure. 
\item 
Consider the situation in Example \ref{network_ex} (vi). We denote the family of kernels $\mu$ from $(\Omega, \fil)$ to $(\bbr_+, \bcal)$ such that $\mu(\cdot, \bbr) \in L^\infty (\Omega; \bbr)$ by $\kcal$. We set $I=L^\infty(\Omega, \bbr^l) \times \kcal$,
$$ Y_k^{(X, \mu)} = \Lambda[ \left\{X_m + \sgn(k_m)\cdot \mu (\cdot,  \mathfrak{B}^m(k_m))\right\}_{m=1,2,\dots, l}  ], \quad (X,\mu)\in I,$$
and define the family of corresponding CVMs by $\ycal = \{ Y^{(X, \mu)}: (X, \mu) \in I  \}$.  
A possible choice of a diversification operator is
$$D_\alpha ((X, \mu), (X', \mu'))_k =  \Lambda[ \left\{\alpha X_m  + (1-\alpha) X'_m+ \sgn(k_m)\cdot (\alpha \mu + (1-\alpha) \mu ' )(\cdot,  \mathfrak{B}^m(k_m))\right\}_{m=1,2,\dots, l}  ]  $$
with $(X, \mu), (X', \mu') \in I$.
Apparently, if $\Lambda$ is concave, then $$D_\alpha ((X, \mu), (X', \mu'))_k  \geq C_\alpha ((X, \mu), (X', \mu'))_k $$ for all $k \in \bbr^l$. Lemma \ref{Propo:rm} implies that $R$ is $D$-quasi-convex,  if $\acal$ is convex.

\end{enumerate}
\item {\bf Conditional systemic risk measures:} As a final example, consider the situation relevant to CoV@R as introduced by \ci{adrian2011covar}.  The setting was explained in Example~\ref{network_ex}(v).   Fix an event $A\in \fcal$ on which we condition. The corresponding conditional  probability measure on $(\Omega, \fil)$ is $P(\cdot | A)$. Setting $\xcal = L^\infty (\Omega, \fil, P(\cdot | A) )$, let $\acal \subseteq \xcal$ be the acceptance set of a distribution-based scalar risk measure $\rho$. This implies that $\rho(X) = \rho(X')$ whenever $X$ and $X'$ have the same distribution under $P(\cdot|A)$. With $I:= L^\infty (\Omega, \fil, P )$, we enumerate $\ycal$ by
$$Y^X_k := k + X , \quad k \in \bbr, \;  X\in I .$$

The special case of CoV@R corresponds to choosing $A$ as in Example~\ref{network_ex}(v) and $\acal\subseteq \xcal$ as the acceptance set of V@R with respect to $P(\cdot|A)$ at level $q$. In this case, one can, of course, not hope that the resulting systemic risk measure has convexity properties, since V@R is not convex. But the conditional setting of \ci{adrian2011covar} can be generalized to other risk measures. 

Suppose now that $\acal\subseteq \xcal$ is convex. The function $k\mapsto Y^X_k$ is concave for any $X$; hence, the corresponding systemic risk measure $R$ has convex values. A natural choice for a diversification operator is 
$$D_\alpha(X,X')_k  = k + \alpha X + (1-\alpha) X'  = C_\alpha(X,X')_k, \quad k \in \bbr, \; \alpha \in [0,1], \; X, X' \in I . $$
 Lemma \ref{Propo:rm}  implies that $R$ is $D$-quasi-convex.
\end{enumerate}
\end{example}

\begin{rem}
As illustrated by various examples, both cash-invariance and diversification may refer to both inputs or outputs of the system. This is not the case for classical scalar monetary risk measures in which these concepts have a unambiguous interpretation. Moreover,  as illustrated in Example~\ref{network_ex}(vi), capital might not simply shift endowments but refer to investments with random payoffs. Our notions of CVMs and systemic risk measures provide a unified framework that includes all these cases. 
\end{rem}

\section{Efficient cash-invariant allocation rules}\label{Sec:Orthant}

Systemic risk is an inherently multi-valued type of risk. In the previous section, we suggested measuring it by the collection of allocations of additional capital that lead to acceptable outcomes. An acceptability criterion can, for example, be chosen by a regulatory authority on the basis of macroprudential objectives. Multi-valued systemic risk measures are then an appropriate tool for the characterization of the underlying downside risk. They also provide a basis for the evaluation and the design of regulatory policies that might go far beyond capital regulation. The goal of regulatory instruments must be to achieve a situation in which actual capital levels are acceptable.

The practical implementation of multi-valued systemic risk measures might, however, be quite challenging. We will thus describe allocation rules that are both efficient and cash-invariant. These will include many risk attribution rules from the literature and provide a general perspective on this issue. As we will demonstrate in Section \ref{Sec:casestudies}, such single-valued allocation rules might in some cases attribute inappropriate risk charges to individual entities. This is an important problem that we point out in the current paper. Since our methodology encompasses many specific constructions from the literature, this issue is not a deficiency of our approach, but a fundamental challenge that we highlight here. 

Any efficient rule for the choice of an allocation of additional capital picks an element $k^*$ of the boundary of $R(Y;k)$. This implies that
\begin{equation}\label{eq:def-orth}
k^* +\bbr^l_+ \subseteq  R(Y;k),
\end{equation} as illustrated in Figure~\ref{fig:orthant_0}. The coordinates of $k^*$ can thus be interpreted as minimal capital requirements for the individual financial firms that lead to an acceptable system. These can be communicated more easily than a systemic risk measure $ R(Y;k)$.  Moreover, the individual choices of acceptable capital levels of the entities within $k^* + \bbr^l_+$ do not cause any externalities to the acceptable capital levels in $k^* + \bbr^l_+$ of the other entities.  In contrast, if the system is required to hold capital allocations in $R(Y;k)$, the feasible region of a single entity $i$ is the section of $R(Y;k)$ at the actual capital $m_{-i} \in \bbr^{l-1}$ of the other entities.
If the other entities modify their actual capital levels inside $R(Y;k)$, the feasible region of entity $i$ will generally change. Note, however, that capital requirements $k^* +\bbr^l_+$ are stricter than systemic risk measures $R(Y;k)$. In this sense, they impose tighter restriction than necessary on the banking system. 

Note, however, that the mentioned externalities are not artificially created  by systemic risk measures. We are dealing with systems, and within these systems externalities are present, since the final outcome is produced by the interaction of the entities of the system. Our systemic risk measure  characterizes exactly those allocations of additional capital that lead to acceptability. At this level, interactions and externalities are \emph{visible}.

\newrgbcolor{db}{0 0 0.5}
\newrgbcolor{lb}{0.3 0.5 1.0}
\definecolor{lg}{gray}{0.7}
\begin{figure}
\centering
\psset{unit=0.4cm}
\begin{pspicture}(0,0)(18,12)
\pspolygon[linewidth=0.0mm,linecolor=white,fillstyle=vlines,hatchcolor=lg,hatchsep=0.4cm,hatchangle=45](.8333,12)(1,10)(2,5)(5,2)(10,1)(18,.5555)(18,12)
\pspolygon[linewidth=0.0mm,linecolor=white,fillstyle=vlines,hatchcolor=lb,hatchsep=0.4cm,hatchangle=-45](3,12)(3,3.3333)(18,3.3333)(18,12)
\pscurve[linewidth=0.6mm,linecolor=black]{-}(.8333,12)(1,10)(2,5)(3,3.3333)(5,2)(10,1)(18,.5555)
\pscircle[linewidth=0.6mm,linecolor=db](3,3.3333){0.2}
\psline[linewidth=0.6mm,linecolor=db]{-}(3,3.3333)(18,3.3333)
\psline[linewidth=0.6mm,linecolor=db]{-}(3,3.3333)(3,12)
\psline[linewidth=0.6mm,linecolor=black]{->}(0,0)(18.3,0)
\psline[linewidth=0.6mm,linecolor=black]{->}(0,0)(0,12.3)
\end{pspicture}
\caption{Illustration of a minimal point $k^*$ of an upper set with the orthant $k^* + \bbr^2_+$ in blue.}
\label{fig:orthant_0}
\end{figure}
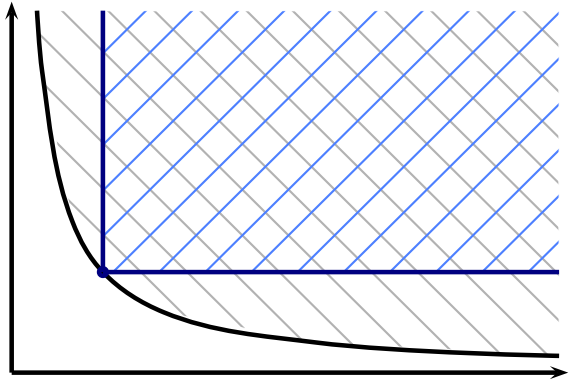

Let us now focus on the construction of efficient cash-invariant allocation rules.  We will always suppose that following { technical} condition holds.  
\begin{assumption}
\label{ass_sec3}
The systemic risk measurement { $R(Y)$} is convex, closed, and non-trivial, i.e.\ { $R(Y) \not\in\{\emptyset,\bbr^l\}$} for every $Y \in \ycal$.
\end{assumption}
\begin{rem} We already provided sufficient conditions for convex and closed values of systemic risk measures in Lemma~\ref{Lem:ci-m} (iii) and Lemma~\ref{Lem:conv-set} .
\end{rem}

In order to be efficient, capital allocations $k^*$ should be as small as possible. In other words, the family  $k^* +\bbr^l_+$ should be a  maximal subset of $R(Y;k)$, as shown in Figure~\ref{fig:orthant_0}. This property is formally slightly stronger than assuming that $k^*$ lies on the boundary of $R(Y;k)$; in fact, it is equivalent to requiring that $ k^*$  is a minimal point of $R(Y;k)$, i.e.
$$ k^* \in \Min R(Y;k):=\{m\in \bbr^l\; | \; (m-\bbr^l_+) \cap R(Y;k)=\{m\}\} .$$
In the following, we define a suitable class of mappings  $k^*$ on $\ycal \times \bbr^l$, called \emph{efficient cash-invariant allocation rules (EAR)}, that choose $k^* \in \Min R(Y;k)$. As we will show, EAR can be characterized as  minimizers of the total weighted cost among all acceptable allocations, if we assign a price of regulatory capital to each entity, see Lemma~\ref{thm:scalarization} below. While typical examples are usually single-valued,   efficient cash-invariant allocation rules  may admit multiple values in non-generic circumstances. We will include a description of these cases in Remark~\ref{rem:scalar}. The mapping $k^*$ is thus set-valued in order to formally allow for these non-generic situations.

\begin{defi}\label{orthant}
 Let $\pcal(\bbr^l)$ be the power set of $\bbr^l$. 
A mapping $k^*: \ycal \times \bbr^l \to \pcal(\bbr^l)$ is called an  \emph{efficient cash-invariant allocation rule (EAR)}  associated with a systemic risk measure $R$, if the following properties are satisfied:
\begin{enumerate}
\item  {\bf Minimal values}: $$k^* (Y;k) \subseteq \Min R(Y;k) \quad(\forall k \in \bbr^l,\;\forall Y \in \ycal).$$
\item  {\bf Convex values}: $$k^1,k^2
\in k^*(Y;k) \;\Rightarrow \; \ga k^1+(1-\ga )k^2 \in k^*(Y;k) \quad (\forall \ga \in [0,1], \; \forall k \in \bbr^l, \; \forall Y\in \ycal).$$ 
\item {\bf Cash-invariance}:
$$k^*(Y;k) + m = k^*(Y;k-m)  \quad(\forall k,m \in \bbr^l, \; \forall Y \in \ycal).$$
\end{enumerate}
\end{defi}

\begin{rem}
Property (i) in Definition~\ref{orthant} was already explained. Property (ii)  guarantees that $k^*(Y;k)$
is single-valued for $Y\in \ycal$ and $k \in \bbr^l$, if $\Min R(Y;k)$ does not contain any line-segments. Property (iii)
of $k^*$ can be interpreted in the same way as the corresponding property of $R$.
\end{rem}

Due to its cash-invariance, it is often sufficient to investigate $k^* (Y;0)$ which we will frequently simply denote by $k^*(Y)$.

\begin{defi}
An  EAR   $k^*: \ycal \times \bbr^l \to \pcal(\bbr^l)$ may possess the following additional properties:
\begin{enumerate}
\setcounter{enumi}{3}
\item {\bf Law-invariance}: 
  Let $Y, Z \in\ycal$.  If for all fixed $k \in \bbr^l$, $Y_k$ and $Z_k$ have the same distribution, then 
 $k^*(Y) = k^*(Z) .$ 
\item {\bf $D$-quasi-convex}: 
Let $\ycal = (Y^i)_{i\in I}$ as in Lemma~\ref{Propo:rm}, and fix $i, j \in I$.  We say that $k^*$ is \emph{$D$-quasi-convex at $(i,j)$} if 
\begin{equation}\label{eq:ortho-quasi}k^*(D_\alpha(i, j)) + \bbr^l_+  \supseteq [k^*(Y^i) + \bbr^l_+]  \cap [k^*(Y^ j) + \bbr^l_+] \end{equation}
for any $\alpha\in [0,1]$. If $k^*$ is $D$-quasi-convex  at any $(i, j) \in I^2$, we call $k^*$ simply \emph{$D$-quasi-convex}.
\end{enumerate}
\end{defi}

\begin{rem}
 Properties (iv) and (v)
of $k^*$ can be interpreted in the same way as the corresponding properties of $R$.
It is, however, possible that the systemic risk measure satisfies this property while a corresponding  EAR  does not. 
\end{rem}

\begin{rem}
If the  EAR  $k^*$ is single-valued, i.e.\  $k^*(Y) \in \bbr^l$ for every $Y \in \ycal$,  then property (v) can be simplified to the more familiar quasi-convexity condition:
\[ k^*(D_\alpha(i,j)) \leq k^*(Y^i) \vee k^*(Y^j),\]
where $k_1 \vee k_2$ denotes the pointwise maximum of $k_1,k_2\in\bbr^l$.
\end{rem}

\begin{proposition}\label{Propo:ortho-conv}
Suppose that Assumption~\ref{ass_sec3} holds, and let $\ycal = (Y^i)_{i\in I}$ as in Lemma~\ref{Propo:rm}.  If $\acal$ is convex 
and $$D_\alpha(i, j)_k\geq C_\alpha(i, j)_k  \quad (\forall k \in \bbr^l),$$ 
for $(i,j) \in I^2$, then there exists some  EAR  $k^*$ which is $D$-quasi-convex at $(i,j)$.  
\end{proposition}

The following lemma provides a characterization and a method for the construction of  EARs.  The EARs  will be characterized as acceptable allocations with minimal weighted costs.  Vectors of prices of regulatory capital are chosen from the dual of the \emph{recession cone} of the systemic risk measure $R(Y)$.  The recession cone of a systemic risk measure is defined by  
\[{\rm recc}R(Y) := \{m \in \bbr^l \; | \; \forall \alpha > 0: \alpha m + R(Y) \subseteq R(Y)\} \supseteq \bbr^l_+ \quad (Y \in \ycal),\] 
which are  the directions within the risk measure in which an arbitrary amount of capital can be added. The positive dual cone is given by 
\[{\rm recc}R(Y)^+ = \{w \in \bbr^l \; | \; \forall m \in {\rm recc}R(Y): w^\T m \geq 0\} \subseteq \bbr^l_+ \quad (Y \in \ycal).\] 
These are the vectors of prices so that by adding capital along the recession cone the total price of regulatory capital does not decrease.  
\begin{lemma}\label{thm:scalarization} 
Suppose that Assumption~\ref{ass_sec3} holds and  ${\rm recc}R(Y) \cap -\bbr^l_+ = \{0\}$  for every $Y \in \ycal$. Let  
$R: \ycal \times \bbr^l \to \pcal(\bbr^l;\bbr^l_+) $ be a systemic risk measure as introduced in Definition~\ref{def:risk}. For $w: \ycal \to \bbr^l_{++} := (0, \infty)^l$ such that  $w(Y) \in {\rm recc} R(Y)^+$,  the set-valued mapping  
\begin{equation}
\label{argminlemma}
\hat k(Y)=
\argmin \left\{\sum_{i = 1}^l w(Y)_i  m_i\; | \; m \in R(Y)\right\} \quad (Y \in \ycal)
\end{equation}
defines an   EAR. All EARs   $k^*$ as defined in Definition~\ref{orthant} are included in   EARs   $\hat k$ of form \eqref{argminlemma}, i.e.\   $k^*(Y)\subseteq \hat k(Y)$ for all $Y \in \ycal$. 
\end{lemma}

As mentioned previously, the vector $w(Y) \in \bbr^l_{++}$ can be interpreted as a vector of prices of the capital requirements of the firms that is chosen by the regulator. For firm $i$ the component $w(Y)_i$ signifies the price that the regulator assigns to one unit of regulatory capital of firm $i$.  EARs  are thus those capital allocations that minimize the total   price-weighted   cost of the capital allocation. 

\begin{rem}\label{rem:scalar}
For $Y\in \ycal$, the value  $\hat k(Y)$  contains more than one element, if and only if  $\Min R(Y)$   contains a line segment that is orthogonal to $w(Y)$.
\end{rem}

\begin{rem} \label{Rem:w}
The weighting vectors $w: \ycal \to \bbr^l_{++}$ have the shift-invariance property $w(Y) = w(Y_{k+\cdot})$ for any $Y \in \ycal$ and $k \in \bbr^l$.  If we apply this to Example~\ref{network_ex}(ii), setting $Y_k^X := \Lambda(X+k)$, $X \in I := L^{\infty}(\Omega; \bbr^n)$ and $k \in \bbr^n$, we obtain that  $w(Y^X)$  depends only on its index, the random vector $X$.
\end{rem}

\begin{rem}
\begin{enumerate}
\item For $Y \in \ycal$ uniqueness of  $\hat k(Y)$   is generic in the following sense, if  $R(Y)$  is not a half-space (implying that  ${\rm recc} R(Y)^+$  is not of dimension $1$) and satisfies the assumptions of Lemma~\ref{thm:scalarization}:  

\noindent Suppose that the direction $w/\sqrt{\sum_{i=1}^l  w_i^2}$ of $w \in \bbr^l_{++}$ is chosen according to a uniform probability measure $\ucal$ on the sphere intersected with $\bbr^l_{++}$.  Then $\argmin \left\{\sum_{i = 1}^l w_i  m_i\; | \; m \in  R(Y)  \right\}$ has cardinality strictly larger than 1 for $w\in {\rm recc}   R(Y)^+$   with $\ucal$-probability 0.
\item If the set of minimal points   $\Min R(Y)$   for  $Y\in \ycal$ does not contain any line segments,  uniqueness is guaranteed.
\end{enumerate}
\end{rem}

While, from a mathematical point of view,  EARs  in Lemma~\ref{thm:scalarization} may be computed for price vectors $w$ depending on the random shock, for practical applications it might sometimes be desirable to choose $w$ simply as constant. In this case, we obtain that the   EAR  can be computed as:  
\begin{equation}\label{eq:orthant}
\hat k (Y)=
\argmin \left\{\sum_{i = 1}^l w_i  m_i\; | \; m \in R(Y)\right\}  \quad\quad\quad (Y \in \ycal)
\end{equation}
for $w\in  \bbr^l_{++} \cap \bigcap_{Y\in  \ycal }{\rm recc} R(Y)^+ $.

\definecolor{lg}{gray}{0.7}
\begin{figure}
\centering
\psset{unit=0.4cm}
\begin{pspicture}(0,0)(18,12)
\pspolygon[linewidth=0.0mm,linecolor=white,fillstyle=vlines,hatchcolor=lg,hatchsep=0.4cm,hatchangle=45](2.1,12.0)(2.4,9.8)(4.0,6.4)(6.6,4.8)(9.2,4.1)(13.8,3.5)(18.0,3.4)(18.0,12.0)
\pspolygon[linewidth=0.0mm,linecolor=white,fillstyle=vlines,hatchcolor=lg,hatchsep=0.4cm,hatchangle=-45](5.4,12.0)(5.5,7.6)(6.3,5.4)(8.3,2.7)(11.0,1.1)(13.5,0.5)(16.8,0.0)(18.0,0.0)(18.0,12.0)
\pscurve[linewidth=0.6mm,linecolor=black]{-}(2.1,12.0)(2.4,9.8)(4.0,6.4)(6.6,4.8)(9.2,4.1)(13.8,3.5)(18.0,3.4)
\pscurve[linewidth=0.6mm,linecolor=black]{-}(5.4,12.0)(5.5,7.6)(6.3,5.4)(8.3,2.7)(11.0,1.1)(13.5,0.4)(16.8,0.0)(18.0,0.0)
\psline[linewidth=0.1mm,linecolor=blue,linestyle=dashed]{-}(0.0,10.35)(10.35,0.0)
\pscircle[linewidth=0.6mm,linecolor=blue](4.2,6.2){0.2}
\psline[linewidth=0.6mm,linecolor=blue,linestyle=dotted](0.0,2.0)(10.0,12.0)
\psline[linewidth=0.1mm,linecolor=red,linestyle=dashed]{-}(0.0,10.95)(10.95,0.0)
\pscircle[linewidth=0.6mm,linecolor=red](8.1,2.9){0.2}
\psline[linewidth=0.6mm,linecolor=red,linestyle=dotted](5.2,0.0)(17.2,12.0)
\psline[linewidth=0.6mm,linecolor=black]{->}(0,0)(18.3,0)
\psline[linewidth=0.6mm,linecolor=black]{->}(0,0)(0,12.3)
\end{pspicture}
\caption{Illustration of  EARs  for two systemic risk measurements.  The tangent and normal lines at the   EAR  $k^*$ are drawn as dashed and dotted lines, respectively.}
\label{fig:orthant}
\end{figure}
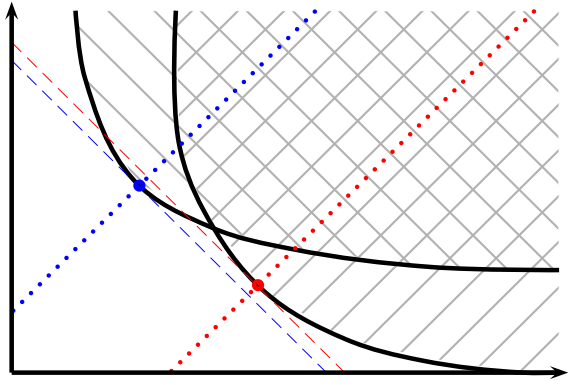

Figure~\ref{fig:orthant} depicts systemic risk measurements together with  EARs  corresponding to a price vector $w = (1,1)^\T$.  The dotted lines indicate the normal vectors, the dashed lines the corresponding tangent vectors. 

\begin{ex}  Let us discuss some examples for the choice of a regulatory price vector $w: \ycal  \to \bbr^l_{++}$. 
\begin{enumerate} 
\item {\bf Fixed price vectors}.
\begin{enumerate}
\item  If $w= (1, \dots, 1)^\T \in \bbr^n$ in Examples \ref{network_ex}(i)-(iii), the regulatory cost of capital of all $n$ firms is equal. In this case, the regulator minimizes the total capital in equation \eqref{eq:orthant} when computing an  EAR.   If the firms are grouped into $l$ entities, as in Example~\ref{network_ex}(iv), with $n_i$ firms in group $i$, then the weighting vector with the same interpretation is given by $w = (n_1,n_2,\dots,n_l)^\T$.
\item  Alternatively, within the context of \ci{EN01} or \ci*{CFS05} introduced in Example \ref{network_ex}(iii), the regulatory price vector $w$ might be defined in terms of ``leverage'', i.e.\ $w_i = \frac 1 {| e_i(0;0 )|}$. The quantity $e_i(0;0)$ signifies the losses to the financial system due to a default of institution $i$ in the case of no recovery.  While the systemic risk measure automatically accounts for the importance of banks in the network, such a choice would select an  EAR  with particularly strict requirements for firms with higher debt levels.
\end{enumerate}
\item {\bf Varying price vectors}.
\begin{enumerate}
\item As we have seen in Remark~\ref{Rem:w}, in the situation of Example \ref{network_ex}(i) and (ii), the regulatory price $w(Y^X)$ is a function of  $X \in L^\infty(\Omega; \bbr^n)$ and could reflect its fluctuations,  e.g.\ $w(Y^X)_i = 1/var(X_i) + c$.  In this case, firms with a high variance of future wealth are assigned a low price for their capital, leading to an  EAR  with a high minimal capital requirement for this firm.  The constant $c \geq 0$ sets a floor on the regulatory prices.
\item Alternatively, also other deviation or (adjusted) downside risk measures could be considered. An example is  $w(Y^X)_i = 1/(\rho_i(X_i) + E[X_i]) + c$ for scalar monetary risk measures $\rho_i$, $i = 1,\dots,n$.  As above, $c \geq 0$ is a floor for the regulatory prices.
\end{enumerate}
\end{enumerate}
\end{ex}

\begin{ex}
 EARs  are closely related to systemic risk measures and  allocation rules  that have been studied previously. 
\begin{enumerate}
\item {\bf Aggregation mechanism insensitive to capital levels.}  The 
systemic risk measures suggested by  \ci*{chen2013axiomatic}, see also \ci*{kromer2013systemic}, possess a structure that is very similar to the proposed EARs, cf.\ Example \ref{network_ex}(i). In the coherent case, both contributions can be embedded into our setting.  Letting $Y^X_k := \Lambda(X) + \sum_{i = 1}^n k_i$ for $k \in \bbr^n$,  the scalar-valued systemic risk measures of \ci*{chen2013axiomatic} can be defined by the minimal total capital required to make the aggregation acceptable, i.e.
\[\rho^{CIM}(Y^X) := \rho(\Lambda(X)) = \min\left\{\sum_{i = 1}^n m_i \; | \; Y^X_m \in \acal\right\}.\]
Additionally, \ci*{chen2013axiomatic} and \ci*{kromer2013systemic} consider a risk allocation mechanism to attribute the value of $\rho^{CIM}$ to the different firms via the solution of a dual representation. This allocation mechanism, which we will denote by $X \mapsto k^{CIM}(X)$, satisfies the full allocation property, i.e.\  $\sum _i k^{CIM}_i(X)  = \rho^{CIM}(Y^X)$. This implies that their allocation is included in an  EAR   with price vector $(1,1,\dots, 1)^\T$:
\[ k^{CIM}(X) \subseteq \hat k(Y^X).\]

In contrast to this paper, \ci*{chen2013axiomatic} and \ci*{kromer2013systemic} do, however, not require the cash-invariance of $\rho$ in all cases that they investigate.
\item {\bf SystRisk.} \ci{brunnermeier2013measuring} suggest a scalar systemic risk measure SystRisk that is based on a generalized version of utility-based shortfall risk (UBSR). In a second step, the scalar risk measure is allocated to individual financial institutions.  This systemic risk measure is closely related to aggregation based systemic risk measures insensitive to capital levels, see Example~\ref{network_ex}(i).  The aggregation function  of \ci{brunnermeier2013measuring},
\[\Lambda(x) := \sum_{i = 1}^n [-\alpha_i x^- + \beta_i (x_i-v_i)^+]\]
for parameters $\alpha,\beta,v \in \bbr^n_+$, treats gains and losses asymmetrically.  
The allocation mechanism $k^*$ of \ci{brunnermeier2013measuring} satisfies the full allocation property, thus   
\[k^*(Y^X) \subseteq \argmin\left\{\sum_{i = 1}^n m_i \; | \; \Lambda(X) + \sum_{i = 1}^n m_i \in \acal \right\},\] 
and is included in an  EAR  with price vector $(1,1,\dots, 1)^\T$. Apart from this, \ci{brunnermeier2013measuring} consider further adjustments to their framework in order to capture additional effects.
\end{enumerate}
\end{ex}

\section{A grid search algorithm}\label{Sec:grid}

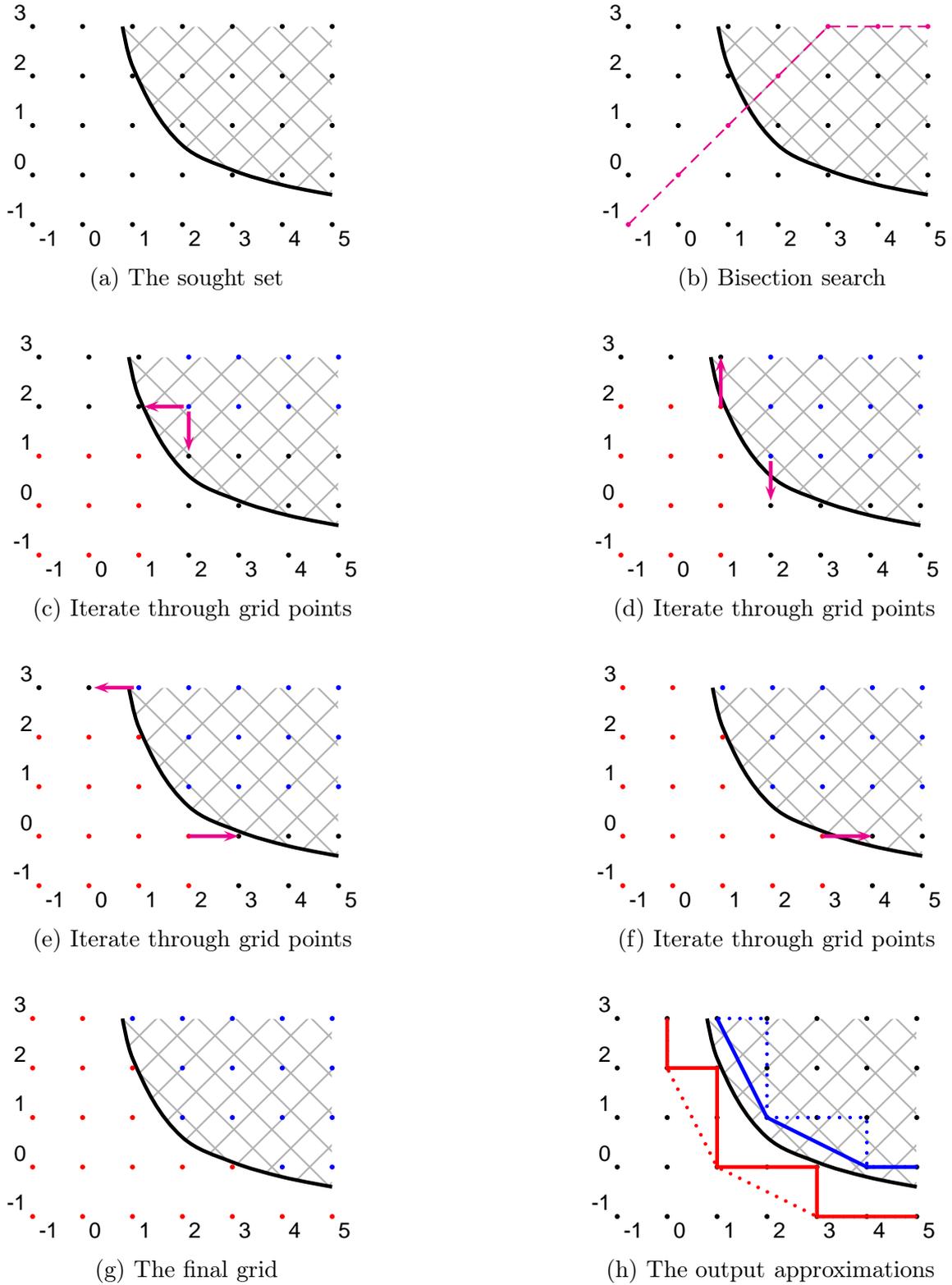
\begin{figure}
\centering
\psset{unit=0.8cm}
\begin{tabular*}{0.9\textwidth}[t]{@{\extracolsep{\fill}} l c r}
\tabpic{
\begin{pspicture}(-1,-1)(5,3)
\psgrid[subgriddiv=0,griddots=1]
\pspolygon[linewidth=0.0mm,linecolor=white,fillstyle=crosshatch,hatchcolor=lg,hatchsep=0.4cm](0.8,3.0)(1.0,2.2)(1.5,1.2)(2.0,0.6)(2.5,0.3)(3.0,0.1)(4.0,-0.2)(5.0,-0.4)(5.0,3.0)
\pscurve[linewidth=0.6mm,linecolor=black]{-}(0.8,3.0)(1.0,2.2)(2.0,0.6)(3.0,0.1)(4.0,-0.2)(5.0,-0.4)
\pscircle[fillstyle=solid,fillcolor=black,linecolor=black](-1.0,-1.0){0.05}
\pscircle[fillstyle=solid,fillcolor=black,linecolor=black](-1.0,0.0){0.05}
\pscircle[fillstyle=solid,fillcolor=black,linecolor=black](-1.0,1.0){0.05}
\pscircle[fillstyle=solid,fillcolor=black,linecolor=black](-1.0,2.0){0.05}
\pscircle[fillstyle=solid,fillcolor=black,linecolor=black](-1.0,3.0){0.05}
\pscircle[fillstyle=solid,fillcolor=black,linecolor=black](0.0,-1.0){0.05}
\pscircle[fillstyle=solid,fillcolor=black,linecolor=black](0.0,0.0){0.05}
\pscircle[fillstyle=solid,fillcolor=black,linecolor=black](0.0,1.0){0.05}
\pscircle[fillstyle=solid,fillcolor=black,linecolor=black](0.0,2.0){0.05}
\pscircle[fillstyle=solid,fillcolor=black,linecolor=black](0.0,3.0){0.05}
\pscircle[fillstyle=solid,fillcolor=black,linecolor=black](1.0,-1.0){0.05}
\pscircle[fillstyle=solid,fillcolor=black,linecolor=black](1.0,0.0){0.05}
\pscircle[fillstyle=solid,fillcolor=black,linecolor=black](1.0,1.0){0.05}
\pscircle[fillstyle=solid,fillcolor=black,linecolor=black](1.0,2.0){0.05}
\pscircle[fillstyle=solid,fillcolor=black,linecolor=black](1.0,3.0){0.05}
\pscircle[fillstyle=solid,fillcolor=black,linecolor=black](2.0,-1.0){0.05}
\pscircle[fillstyle=solid,fillcolor=black,linecolor=black](2.0,0.0){0.05}
\pscircle[fillstyle=solid,fillcolor=black,linecolor=black](2.0,1.0){0.05}
\pscircle[fillstyle=solid,fillcolor=black,linecolor=black](2.0,2.0){0.05}
\pscircle[fillstyle=solid,fillcolor=black,linecolor=black](2.0,3.0){0.05}
\pscircle[fillstyle=solid,fillcolor=black,linecolor=black](3.0,-1.0){0.05}
\pscircle[fillstyle=solid,fillcolor=black,linecolor=black](3.0,0.0){0.05}
\pscircle[fillstyle=solid,fillcolor=black,linecolor=black](3.0,1.0){0.05}
\pscircle[fillstyle=solid,fillcolor=black,linecolor=black](3.0,2.0){0.05}
\pscircle[fillstyle=solid,fillcolor=black,linecolor=black](3.0,3.0){0.05}
\pscircle[fillstyle=solid,fillcolor=black,linecolor=black](4.0,-1.0){0.05}
\pscircle[fillstyle=solid,fillcolor=black,linecolor=black](4.0,0.0){0.05}
\pscircle[fillstyle=solid,fillcolor=black,linecolor=black](4.0,1.0){0.05}
\pscircle[fillstyle=solid,fillcolor=black,linecolor=black](4.0,2.0){0.05}
\pscircle[fillstyle=solid,fillcolor=black,linecolor=black](4.0,3.0){0.05}
\pscircle[fillstyle=solid,fillcolor=black,linecolor=black](5.0,-1.0){0.05}
\pscircle[fillstyle=solid,fillcolor=black,linecolor=black](5.0,0.0){0.05}
\pscircle[fillstyle=solid,fillcolor=black,linecolor=black](5.0,1.0){0.05}
\pscircle[fillstyle=solid,fillcolor=black,linecolor=black](5.0,2.0){0.05}
\pscircle[fillstyle=solid,fillcolor=black,linecolor=black](5.0,3.0){0.05}
\end{pspicture}
}{The sought set}
\tabpic{
\begin{pspicture}(-1.0,-1.0)(5.0,3.0)
\psgrid[subgriddiv=0,griddots=1]
\pspolygon[linewidth=0.0mm,linecolor=white,fillstyle=crosshatch,hatchcolor=lg,hatchsep=0.4cm](0.8,3.0)(1.0,2.2)(1.5,1.2)(2.0,0.6)(2.5,0.3)(3.0,0.1)(4.0,-0.2)(5.0,-0.4)(5.0,3.0)
\pscurve[linewidth=0.6mm,linecolor=black]{-}(0.8,3.0)(1.0,2.2)(2.0,0.6)(3.0,0.1)(4.0,-0.2)(5.0,-0.4)
\pscircle[fillstyle=solid,fillcolor=magenta,linecolor=magenta](-1.0,-1.0){0.05}
\pscircle[fillstyle=solid,fillcolor=black,linecolor=black](-1.0,0.0){0.05}
\pscircle[fillstyle=solid,fillcolor=black,linecolor=black](-1.0,1.0){0.05}
\pscircle[fillstyle=solid,fillcolor=black,linecolor=black](-1.0,2.0){0.05}
\pscircle[fillstyle=solid,fillcolor=black,linecolor=black](-1.0,3.0){0.05}
\pscircle[fillstyle=solid,fillcolor=black,linecolor=black](0.0,-1.0){0.05}
\pscircle[fillstyle=solid,fillcolor=magenta,linecolor=magenta](0.0,0.0){0.05}
\pscircle[fillstyle=solid,fillcolor=black,linecolor=black](0.0,1.0){0.05}
\pscircle[fillstyle=solid,fillcolor=black,linecolor=black](0.0,2.0){0.05}
\pscircle[fillstyle=solid,fillcolor=black,linecolor=black](0.0,3.0){0.05}
\pscircle[fillstyle=solid,fillcolor=black,linecolor=black](1.0,-1.0){0.05}
\pscircle[fillstyle=solid,fillcolor=black,linecolor=black](1.0,0.0){0.05}
\pscircle[fillstyle=solid,fillcolor=magenta,linecolor=magenta](1.0,1.0){0.05}
\pscircle[fillstyle=solid,fillcolor=black,linecolor=black](1.0,2.0){0.05}
\pscircle[fillstyle=solid,fillcolor=black,linecolor=black](1.0,3.0){0.05}
\pscircle[fillstyle=solid,fillcolor=black,linecolor=black](2.0,-1.0){0.05}
\pscircle[fillstyle=solid,fillcolor=black,linecolor=black](2.0,0.0){0.05}
\pscircle[fillstyle=solid,fillcolor=black,linecolor=black](2.0,1.0){0.05}
\pscircle[fillstyle=solid,fillcolor=magenta,linecolor=magenta](2.0,2.0){0.05}
\pscircle[fillstyle=solid,fillcolor=black,linecolor=black](2.0,3.0){0.05}
\pscircle[fillstyle=solid,fillcolor=black,linecolor=black](3.0,-1.0){0.05}
\pscircle[fillstyle=solid,fillcolor=black,linecolor=black](3.0,0.0){0.05}
\pscircle[fillstyle=solid,fillcolor=black,linecolor=black](3.0,1.0){0.05}
\pscircle[fillstyle=solid,fillcolor=black,linecolor=black](3.0,2.0){0.05}
\pscircle[fillstyle=solid,fillcolor=magenta,linecolor=magenta](3.0,3.0){0.05}
\pscircle[fillstyle=solid,fillcolor=black,linecolor=black](4.0,-1.0){0.05}
\pscircle[fillstyle=solid,fillcolor=black,linecolor=black](4.0,0.0){0.05}
\pscircle[fillstyle=solid,fillcolor=black,linecolor=black](4.0,1.0){0.05}
\pscircle[fillstyle=solid,fillcolor=black,linecolor=black](4.0,2.0){0.05}
\pscircle[fillstyle=solid,fillcolor=magenta,linecolor=magenta](4.0,3.0){0.05}
\pscircle[fillstyle=solid,fillcolor=black,linecolor=black](5.0,-1.0){0.05}
\pscircle[fillstyle=solid,fillcolor=black,linecolor=black](5.0,0.0){0.05}
\pscircle[fillstyle=solid,fillcolor=black,linecolor=black](5.0,1.0){0.05}
\pscircle[fillstyle=solid,fillcolor=black,linecolor=black](5.0,2.0){0.05}
\pscircle[fillstyle=solid,fillcolor=magenta,linecolor=magenta](5.0,3.0){0.05}
\psline[linewidth=0.3mm,linecolor=magenta,linestyle=dashed]{-}(-1.0,-1.0)(0.0,0.0)(1.0,1.0)(2.0,2.0)(3.0,3.0)(4.0,3.0)(5.0,3.0)
\end{pspicture}
}{Bisection search}
\tabpic{
\begin{pspicture}(-1.0,-1.0)(5.0,3.0)
\psgrid[subgriddiv=0,griddots=1]
\pspolygon[linewidth=0.0mm,linecolor=white,fillstyle=crosshatch,hatchcolor=lg,hatchsep=0.4cm](0.8,3.0)(1.0,2.2)(1.5,1.2)(2.0,0.6)(2.5,0.3)(3.0,0.1)(4.0,-0.2)(5.0,-0.4)(5.0,3.0)
\pscurve[linewidth=0.6mm,linecolor=black]{-}(0.8,3.0)(1.0,2.2)(2.0,0.6)(3.0,0.1)(4.0,-0.2)(5.0,-0.4)
\pscircle[fillstyle=solid,fillcolor=red,linecolor=red](-1.0,-1.0){0.05}
\pscircle[fillstyle=solid,fillcolor=red,linecolor=red](-1.0,0.0){0.05}
\pscircle[fillstyle=solid,fillcolor=red,linecolor=red](-1.0,1.0){0.05}
\pscircle[fillstyle=solid,fillcolor=black,linecolor=black](-1.0,2.0){0.05}
\pscircle[fillstyle=solid,fillcolor=black,linecolor=black](-1.0,3.0){0.05}
\pscircle[fillstyle=solid,fillcolor=red,linecolor=red](0.0,-1.0){0.05}
\pscircle[fillstyle=solid,fillcolor=red,linecolor=red](0.0,0.0){0.05}
\pscircle[fillstyle=solid,fillcolor=red,linecolor=red](0.0,1.0){0.05}
\pscircle[fillstyle=solid,fillcolor=black,linecolor=black](0.0,2.0){0.05}
\pscircle[fillstyle=solid,fillcolor=black,linecolor=black](0.0,3.0){0.05}
\pscircle[fillstyle=solid,fillcolor=red,linecolor=red](1.0,-1.0){0.05}
\pscircle[fillstyle=solid,fillcolor=red,linecolor=red](1.0,0.0){0.05}
\pscircle[fillstyle=solid,fillcolor=red,linecolor=red](1.0,1.0){0.05}
\pscircle[fillstyle=solid,fillcolor=black,linecolor=black](1.0,2.0){0.05}
\pscircle[fillstyle=solid,fillcolor=black,linecolor=black](1.0,3.0){0.05}
\pscircle[fillstyle=solid,fillcolor=black,linecolor=black](2.0,-1.0){0.05}
\pscircle[fillstyle=solid,fillcolor=black,linecolor=black](2.0,0.0){0.05}
\pscircle[fillstyle=solid,fillcolor=black,linecolor=black](2.0,1.0){0.05}
\pscircle[fillstyle=solid,fillcolor=blue,linecolor=blue](2.0,2.0){0.05}
\pscircle[fillstyle=solid,fillcolor=blue,linecolor=blue](2.0,3.0){0.05}
\pscircle[fillstyle=solid,fillcolor=black,linecolor=black](3.0,-1.0){0.05}
\pscircle[fillstyle=solid,fillcolor=black,linecolor=black](3.0,0.0){0.05}
\pscircle[fillstyle=solid,fillcolor=black,linecolor=black](3.0,1.0){0.05}
\pscircle[fillstyle=solid,fillcolor=blue,linecolor=blue](3.0,2.0){0.05}
\pscircle[fillstyle=solid,fillcolor=blue,linecolor=blue](3.0,3.0){0.05}
\pscircle[fillstyle=solid,fillcolor=black,linecolor=black](4.0,-1.0){0.05}
\pscircle[fillstyle=solid,fillcolor=black,linecolor=black](4.0,0.0){0.05}
\pscircle[fillstyle=solid,fillcolor=black,linecolor=black](4.0,1.0){0.05}
\pscircle[fillstyle=solid,fillcolor=blue,linecolor=blue](4.0,2.0){0.05}
\pscircle[fillstyle=solid,fillcolor=blue,linecolor=blue](4.0,3.0){0.05}
\pscircle[fillstyle=solid,fillcolor=black,linecolor=black](5.0,-1.0){0.05}
\pscircle[fillstyle=solid,fillcolor=black,linecolor=black](5.0,0.0){0.05}
\pscircle[fillstyle=solid,fillcolor=black,linecolor=black](5.0,1.0){0.05}
\pscircle[fillstyle=solid,fillcolor=blue,linecolor=blue](5.0,2.0){0.05}
\pscircle[fillstyle=solid,fillcolor=blue,linecolor=blue](5.0,3.0){0.05}
\psline[linewidth=0.6mm,linecolor=magenta]{->}(1.9,2.0)(1.1,2.0)
\psline[linewidth=0.6mm,linecolor=magenta]{->}(2.0,1.9)(2.0,1.1)
\end{pspicture}
}{Iterate through grid points}
\tabpic{
\begin{pspicture}(-1.0,-1.0)(5.0,3.0)
\psgrid[subgriddiv=0,griddots=1]
\pspolygon[linewidth=0.0mm,linecolor=white,fillstyle=crosshatch,hatchcolor=lg,hatchsep=0.4cm](0.8,3.0)(1.0,2.2)(1.5,1.2)(2.0,0.6)(2.5,0.3)(3.0,0.1)(4.0,-0.2)(5.0,-0.4)(5.0,3.0)
\pscurve[linewidth=0.6mm,linecolor=black]{-}(0.8,3.0)(1.0,2.2)(2.0,0.6)(3.0,0.1)(4.0,-0.2)(5.0,-0.4)
\pscircle[fillstyle=solid,fillcolor=red,linecolor=red](-1.0,-1.0){0.05}
\pscircle[fillstyle=solid,fillcolor=red,linecolor=red](-1.0,0.0){0.05}
\pscircle[fillstyle=solid,fillcolor=red,linecolor=red](-1.0,1.0){0.05}
\pscircle[fillstyle=solid,fillcolor=red,linecolor=red](-1.0,2.0){0.05}
\pscircle[fillstyle=solid,fillcolor=black,linecolor=black](-1.0,3.0){0.05}
\pscircle[fillstyle=solid,fillcolor=red,linecolor=red](0.0,-1.0){0.05}
\pscircle[fillstyle=solid,fillcolor=red,linecolor=red](0.0,0.0){0.05}
\pscircle[fillstyle=solid,fillcolor=red,linecolor=red](0.0,1.0){0.05}
\pscircle[fillstyle=solid,fillcolor=red,linecolor=red](0.0,2.0){0.05}
\pscircle[fillstyle=solid,fillcolor=black,linecolor=black](0.0,3.0){0.05}
\pscircle[fillstyle=solid,fillcolor=red,linecolor=red](1.0,-1.0){0.05}
\pscircle[fillstyle=solid,fillcolor=red,linecolor=red](1.0,0.0){0.05}
\pscircle[fillstyle=solid,fillcolor=red,linecolor=red](1.0,1.0){0.05}
\pscircle[fillstyle=solid,fillcolor=red,linecolor=red](1.0,2.0){0.05}
\pscircle[fillstyle=solid,fillcolor=black,linecolor=black](1.0,3.0){0.05}
\pscircle[fillstyle=solid,fillcolor=black,linecolor=black](2.0,-1.0){0.05}
\pscircle[fillstyle=solid,fillcolor=black,linecolor=black](2.0,0.0){0.05}
\pscircle[fillstyle=solid,fillcolor=blue,linecolor=blue](2.0,1.0){0.05}
\pscircle[fillstyle=solid,fillcolor=blue,linecolor=blue](2.0,2.0){0.05}
\pscircle[fillstyle=solid,fillcolor=blue,linecolor=blue](2.0,3.0){0.05}
\pscircle[fillstyle=solid,fillcolor=black,linecolor=black](3.0,-1.0){0.05}
\pscircle[fillstyle=solid,fillcolor=black,linecolor=black](3.0,0.0){0.05}
\pscircle[fillstyle=solid,fillcolor=blue,linecolor=blue](3.0,1.0){0.05}
\pscircle[fillstyle=solid,fillcolor=blue,linecolor=blue](3.0,2.0){0.05}
\pscircle[fillstyle=solid,fillcolor=blue,linecolor=blue](3.0,3.0){0.05}
\pscircle[fillstyle=solid,fillcolor=black,linecolor=black](4.0,-1.0){0.05}
\pscircle[fillstyle=solid,fillcolor=black,linecolor=black](4.0,0.0){0.05}
\pscircle[fillstyle=solid,fillcolor=blue,linecolor=blue](4.0,1.0){0.05}
\pscircle[fillstyle=solid,fillcolor=blue,linecolor=blue](4.0,2.0){0.05}
\pscircle[fillstyle=solid,fillcolor=blue,linecolor=blue](4.0,3.0){0.05}
\pscircle[fillstyle=solid,fillcolor=black,linecolor=black](5.0,-1.0){0.05}
\pscircle[fillstyle=solid,fillcolor=black,linecolor=black](5.0,0.0){0.05}
\pscircle[fillstyle=solid,fillcolor=blue,linecolor=blue](5.0,1.0){0.05}
\pscircle[fillstyle=solid,fillcolor=blue,linecolor=blue](5.0,2.0){0.05}
\pscircle[fillstyle=solid,fillcolor=blue,linecolor=blue](5.0,3.0){0.05}
\psline[linewidth=0.6mm,linecolor=magenta]{->}(1.0,2.0)(1.0,3.0)
\psline[linewidth=0.6mm,linecolor=magenta]{->}(2.0,0.9)(2.0,0.1)
\end{pspicture}
}{Iterate through grid points}
\tabpic{
\begin{pspicture}(-1.0,-1.0)(5.0,3.0)
\psgrid[subgriddiv=0,griddots=1]
\pspolygon[linewidth=0.0mm,linecolor=white,fillstyle=crosshatch,hatchcolor=lg,hatchsep=0.4cm](0.8,3.0)(1.0,2.2)(1.5,1.2)(2.0,0.6)(2.5,0.3)(3.0,0.1)(4.0,-0.2)(5.0,-0.4)(5.0,3.0)
\pscurve[linewidth=0.6mm,linecolor=black]{-}(0.8,3.0)(1.0,2.2)(2.0,0.6)(3.0,0.1)(4.0,-0.2)(5.0,-0.4)
\pscircle[fillstyle=solid,fillcolor=red,linecolor=red](-1.0,-1.0){0.05}
\pscircle[fillstyle=solid,fillcolor=red,linecolor=red](-1.0,0.0){0.05}
\pscircle[fillstyle=solid,fillcolor=red,linecolor=red](-1.0,1.0){0.05}
\pscircle[fillstyle=solid,fillcolor=red,linecolor=red](-1.0,2.0){0.05}
\pscircle[fillstyle=solid,fillcolor=black,linecolor=black](-1.0,3.0){0.05}
\pscircle[fillstyle=solid,fillcolor=red,linecolor=red](0.0,-1.0){0.05}
\pscircle[fillstyle=solid,fillcolor=red,linecolor=red](0.0,0.0){0.05}
\pscircle[fillstyle=solid,fillcolor=red,linecolor=red](0.0,1.0){0.05}
\pscircle[fillstyle=solid,fillcolor=red,linecolor=red](0.0,2.0){0.05}
\pscircle[fillstyle=solid,fillcolor=black,linecolor=black](0.0,3.0){0.05}
\pscircle[fillstyle=solid,fillcolor=red,linecolor=red](1.0,-1.0){0.05}
\pscircle[fillstyle=solid,fillcolor=red,linecolor=red](1.0,0.0){0.05}
\pscircle[fillstyle=solid,fillcolor=red,linecolor=red](1.0,1.0){0.05}
\pscircle[fillstyle=solid,fillcolor=red,linecolor=red](1.0,2.0){0.05}
\pscircle[fillstyle=solid,fillcolor=blue,linecolor=blue](1.0,3.0){0.05}
\pscircle[fillstyle=solid,fillcolor=red,linecolor=red](2.0,-1.0){0.05}
\pscircle[fillstyle=solid,fillcolor=red,linecolor=red](2.0,0.0){0.05}
\pscircle[fillstyle=solid,fillcolor=blue,linecolor=blue](2.0,1.0){0.05}
\pscircle[fillstyle=solid,fillcolor=blue,linecolor=blue](2.0,2.0){0.05}
\pscircle[fillstyle=solid,fillcolor=blue,linecolor=blue](2.0,3.0){0.05}
\pscircle[fillstyle=solid,fillcolor=black,linecolor=black](3.0,-1.0){0.05}
\pscircle[fillstyle=solid,fillcolor=black,linecolor=black](3.0,0.0){0.05}
\pscircle[fillstyle=solid,fillcolor=blue,linecolor=blue](3.0,1.0){0.05}
\pscircle[fillstyle=solid,fillcolor=blue,linecolor=blue](3.0,2.0){0.05}
\pscircle[fillstyle=solid,fillcolor=blue,linecolor=blue](3.0,3.0){0.05}
\pscircle[fillstyle=solid,fillcolor=black,linecolor=black](4.0,-1.0){0.05}
\pscircle[fillstyle=solid,fillcolor=black,linecolor=black](4.0,0.0){0.05}
\pscircle[fillstyle=solid,fillcolor=blue,linecolor=blue](4.0,1.0){0.05}
\pscircle[fillstyle=solid,fillcolor=blue,linecolor=blue](4.0,2.0){0.05}
\pscircle[fillstyle=solid,fillcolor=blue,linecolor=blue](4.0,3.0){0.05}
\pscircle[fillstyle=solid,fillcolor=black,linecolor=black](5.0,-1.0){0.05}
\pscircle[fillstyle=solid,fillcolor=black,linecolor=black](5.0,0.0){0.05}
\pscircle[fillstyle=solid,fillcolor=blue,linecolor=blue](5.0,1.0){0.05}
\pscircle[fillstyle=solid,fillcolor=blue,linecolor=blue](5.0,2.0){0.05}
\pscircle[fillstyle=solid,fillcolor=blue,linecolor=blue](5.0,3.0){0.05}
\psline[linewidth=0.6mm,linecolor=magenta]{->}(0.9,3.0)(0.1,3.0)
\psline[linewidth=0.6mm,linecolor=magenta]{->}(2.0,0.0)(3.0,0.0)
\end{pspicture}
}{Iterate through grid points}

\tabpic{
\begin{pspicture}(-1.0,-1.0)(5.0,3.0)
\psgrid[subgriddiv=0,griddots=1]
\pspolygon[linewidth=0.0mm,linecolor=white,fillstyle=crosshatch,hatchcolor=lg,hatchsep=0.4cm](0.8,3.0)(1.0,2.2)(1.5,1.2)(2.0,0.6)(2.5,0.3)(3.0,0.1)(4.0,-0.2)(5.0,-0.4)(5.0,3.0)
\pscurve[linewidth=0.6mm,linecolor=black]{-}(0.8,3.0)(1.0,2.2)(2.0,0.6)(3.0,0.1)(4.0,-0.2)(5.0,-0.4)
\pscircle[fillstyle=solid,fillcolor=red,linecolor=red](-1.0,-1.0){0.05}
\pscircle[fillstyle=solid,fillcolor=red,linecolor=red](-1.0,0.0){0.05}
\pscircle[fillstyle=solid,fillcolor=red,linecolor=red](-1.0,1.0){0.05}
\pscircle[fillstyle=solid,fillcolor=red,linecolor=red](-1.0,2.0){0.05}
\pscircle[fillstyle=solid,fillcolor=red,linecolor=red](-1.0,3.0){0.05}
\pscircle[fillstyle=solid,fillcolor=red,linecolor=red](0.0,-1.0){0.05}
\pscircle[fillstyle=solid,fillcolor=red,linecolor=red](0.0,0.0){0.05}
\pscircle[fillstyle=solid,fillcolor=red,linecolor=red](0.0,1.0){0.05}
\pscircle[fillstyle=solid,fillcolor=red,linecolor=red](0.0,2.0){0.05}
\pscircle[fillstyle=solid,fillcolor=red,linecolor=red](0.0,3.0){0.05}
\pscircle[fillstyle=solid,fillcolor=red,linecolor=red](1.0,-1.0){0.05}
\pscircle[fillstyle=solid,fillcolor=red,linecolor=red](1.0,0.0){0.05}
\pscircle[fillstyle=solid,fillcolor=red,linecolor=red](1.0,1.0){0.05}
\pscircle[fillstyle=solid,fillcolor=red,linecolor=red](1.0,2.0){0.05}
\pscircle[fillstyle=solid,fillcolor=blue,linecolor=blue](1.0,3.0){0.05}
\pscircle[fillstyle=solid,fillcolor=red,linecolor=red](2.0,-1.0){0.05}
\pscircle[fillstyle=solid,fillcolor=red,linecolor=red](2.0,0.0){0.05}
\pscircle[fillstyle=solid,fillcolor=blue,linecolor=blue](2.0,1.0){0.05}
\pscircle[fillstyle=solid,fillcolor=blue,linecolor=blue](2.0,2.0){0.05}
\pscircle[fillstyle=solid,fillcolor=blue,linecolor=blue](2.0,3.0){0.05}
\pscircle[fillstyle=solid,fillcolor=red,linecolor=red](3.0,-1.0){0.05}
\pscircle[fillstyle=solid,fillcolor=red,linecolor=red](3.0,0.0){0.05}
\pscircle[fillstyle=solid,fillcolor=blue,linecolor=blue](3.0,1.0){0.05}
\pscircle[fillstyle=solid,fillcolor=blue,linecolor=blue](3.0,2.0){0.05}
\pscircle[fillstyle=solid,fillcolor=blue,linecolor=blue](3.0,3.0){0.05}
\pscircle[fillstyle=solid,fillcolor=black,linecolor=black](4.0,-1.0){0.05}
\pscircle[fillstyle=solid,fillcolor=black,linecolor=black](4.0,0.0){0.05}
\pscircle[fillstyle=solid,fillcolor=blue,linecolor=blue](4.0,1.0){0.05}
\pscircle[fillstyle=solid,fillcolor=blue,linecolor=blue](4.0,2.0){0.05}
\pscircle[fillstyle=solid,fillcolor=blue,linecolor=blue](4.0,3.0){0.05}
\pscircle[fillstyle=solid,fillcolor=black,linecolor=black](5.0,-1.0){0.05}
\pscircle[fillstyle=solid,fillcolor=black,linecolor=black](5.0,0.0){0.05}
\pscircle[fillstyle=solid,fillcolor=blue,linecolor=blue](5.0,1.0){0.05}
\pscircle[fillstyle=solid,fillcolor=blue,linecolor=blue](5.0,2.0){0.05}
\pscircle[fillstyle=solid,fillcolor=blue,linecolor=blue](5.0,3.0){0.05}
\psline[linewidth=0.6mm,linecolor=magenta]{->}(3.0,0.0)(4.0,0.0)
\end{pspicture}
}{Iterate through grid points}
\tabpic{
\begin{pspicture}(-1.0,-1.0)(5.0,3.0)
\psgrid[subgriddiv=0,griddots=1]
\pspolygon[linewidth=0.0mm,linecolor=white,fillstyle=crosshatch,hatchcolor=lg,hatchsep=0.4cm](0.8,3.0)(1.0,2.2)(1.5,1.2)(2.0,0.6)(2.5,0.3)(3.0,0.1)(4.0,-0.2)(5.0,-0.4)(5.0,3.0)
\pscurve[linewidth=0.6mm,linecolor=black]{-}(0.8,3.0)(1.0,2.2)(2.0,0.6)(3.0,0.1)(4.0,-0.2)(5.0,-0.4)
\pscircle[fillstyle=solid,fillcolor=red,linecolor=red](-1.0,-1.0){0.05}
\pscircle[fillstyle=solid,fillcolor=red,linecolor=red](-1.0,0.0){0.05}
\pscircle[fillstyle=solid,fillcolor=red,linecolor=red](-1.0,1.0){0.05}
\pscircle[fillstyle=solid,fillcolor=red,linecolor=red](-1.0,2.0){0.05}
\pscircle[fillstyle=solid,fillcolor=red,linecolor=red](-1.0,3.0){0.05}
\pscircle[fillstyle=solid,fillcolor=red,linecolor=red](0.0,-1.0){0.05}
\pscircle[fillstyle=solid,fillcolor=red,linecolor=red](0.0,0.0){0.05}
\pscircle[fillstyle=solid,fillcolor=red,linecolor=red](0.0,1.0){0.05}
\pscircle[fillstyle=solid,fillcolor=red,linecolor=red](0.0,2.0){0.05}
\pscircle[fillstyle=solid,fillcolor=red,linecolor=red](0.0,3.0){0.05}
\pscircle[fillstyle=solid,fillcolor=red,linecolor=red](1.0,-1.0){0.05}
\pscircle[fillstyle=solid,fillcolor=red,linecolor=red](1.0,0.0){0.05}
\pscircle[fillstyle=solid,fillcolor=red,linecolor=red](1.0,1.0){0.05}
\pscircle[fillstyle=solid,fillcolor=red,linecolor=red](1.0,2.0){0.05}
\pscircle[fillstyle=solid,fillcolor=blue,linecolor=blue](1.0,3.0){0.05}
\pscircle[fillstyle=solid,fillcolor=red,linecolor=red](2.0,-1.0){0.05}
\pscircle[fillstyle=solid,fillcolor=red,linecolor=red](2.0,0.0){0.05}
\pscircle[fillstyle=solid,fillcolor=blue,linecolor=blue](2.0,1.0){0.05}
\pscircle[fillstyle=solid,fillcolor=blue,linecolor=blue](2.0,2.0){0.05}
\pscircle[fillstyle=solid,fillcolor=blue,linecolor=blue](2.0,3.0){0.05}
\pscircle[fillstyle=solid,fillcolor=red,linecolor=red](3.0,-1.0){0.05}
\pscircle[fillstyle=solid,fillcolor=red,linecolor=red](3.0,0.0){0.05}
\pscircle[fillstyle=solid,fillcolor=blue,linecolor=blue](3.0,1.0){0.05}
\pscircle[fillstyle=solid,fillcolor=blue,linecolor=blue](3.0,2.0){0.05}
\pscircle[fillstyle=solid,fillcolor=blue,linecolor=blue](3.0,3.0){0.05}
\pscircle[fillstyle=solid,fillcolor=red,linecolor=red](4.0,-1.0){0.05}
\pscircle[fillstyle=solid,fillcolor=blue,linecolor=blue](4.0,0.0){0.05}
\pscircle[fillstyle=solid,fillcolor=blue,linecolor=blue](4.0,1.0){0.05}
\pscircle[fillstyle=solid,fillcolor=blue,linecolor=blue](4.0,2.0){0.05}
\pscircle[fillstyle=solid,fillcolor=blue,linecolor=blue](4.0,3.0){0.05}
\pscircle[fillstyle=solid,fillcolor=red,linecolor=red](5.0,-1.0){0.05}
\pscircle[fillstyle=solid,fillcolor=blue,linecolor=blue](5.0,0.0){0.05}
\pscircle[fillstyle=solid,fillcolor=blue,linecolor=blue](5.0,1.0){0.05}
\pscircle[fillstyle=solid,fillcolor=blue,linecolor=blue](5.0,2.0){0.05}
\pscircle[fillstyle=solid,fillcolor=blue,linecolor=blue](5.0,3.0){0.05}
\end{pspicture}
}{The final grid}
\tabpic{
\begin{pspicture}(-1.0,-1.0)(5.0,3.0)
\psgrid[subgriddiv=0,griddots=1]
\pspolygon[linewidth=0.0mm,linecolor=white,fillstyle=crosshatch,hatchcolor=lg,hatchsep=0.4cm](0.8,3.0)(1.0,2.2)(1.5,1.2)(2.0,0.6)(2.5,0.3)(3.0,0.1)(4.0,-0.2)(5.0,-0.4)(5.0,3.0)
\pscurve[linewidth=0.6mm,linecolor=black]{-}(0.8,3.0)(1.0,2.2)(2.0,0.6)(3.0,0.1)(4.0,-0.2)(5.0,-0.4)
\pscircle[fillstyle=solid,fillcolor=black,linecolor=black](-1.0,-1.0){0.05}
\pscircle[fillstyle=solid,fillcolor=black,linecolor=black](-1.0,0.0){0.05}
\pscircle[fillstyle=solid,fillcolor=black,linecolor=black](-1.0,1.0){0.05}
\pscircle[fillstyle=solid,fillcolor=black,linecolor=black](-1.0,2.0){0.05}
\pscircle[fillstyle=solid,fillcolor=black,linecolor=black](-1.0,3.0){0.05}
\pscircle[fillstyle=solid,fillcolor=black,linecolor=black](0.0,-1.0){0.05}
\pscircle[fillstyle=solid,fillcolor=black,linecolor=black](0.0,0.0){0.05}
\pscircle[fillstyle=solid,fillcolor=black,linecolor=black](0.0,1.0){0.05}
\pscircle[fillstyle=solid,fillcolor=black,linecolor=black](0.0,2.0){0.05}
\pscircle[fillstyle=solid,fillcolor=black,linecolor=black](0.0,3.0){0.05}
\pscircle[fillstyle=solid,fillcolor=black,linecolor=black](1.0,-1.0){0.05}
\pscircle[fillstyle=solid,fillcolor=black,linecolor=black](1.0,0.0){0.05}
\pscircle[fillstyle=solid,fillcolor=black,linecolor=black](1.0,1.0){0.05}
\pscircle[fillstyle=solid,fillcolor=black,linecolor=black](1.0,2.0){0.05}
\pscircle[fillstyle=solid,fillcolor=black,linecolor=black](1.0,3.0){0.05}
\pscircle[fillstyle=solid,fillcolor=black,linecolor=black](2.0,-1.0){0.05}
\pscircle[fillstyle=solid,fillcolor=black,linecolor=black](2.0,0.0){0.05}
\pscircle[fillstyle=solid,fillcolor=black,linecolor=black](2.0,1.0){0.05}
\pscircle[fillstyle=solid,fillcolor=black,linecolor=black](2.0,2.0){0.05}
\pscircle[fillstyle=solid,fillcolor=black,linecolor=black](2.0,3.0){0.05}
\pscircle[fillstyle=solid,fillcolor=black,linecolor=black](3.0,-1.0){0.05}
\pscircle[fillstyle=solid,fillcolor=black,linecolor=black](3.0,0.0){0.05}
\pscircle[fillstyle=solid,fillcolor=black,linecolor=black](3.0,1.0){0.05}
\pscircle[fillstyle=solid,fillcolor=black,linecolor=black](3.0,2.0){0.05}
\pscircle[fillstyle=solid,fillcolor=black,linecolor=black](3.0,3.0){0.05}
\pscircle[fillstyle=solid,fillcolor=black,linecolor=black](4.0,-1.0){0.05}
\pscircle[fillstyle=solid,fillcolor=black,linecolor=black](4.0,0.0){0.05}
\pscircle[fillstyle=solid,fillcolor=black,linecolor=black](4.0,1.0){0.05}
\pscircle[fillstyle=solid,fillcolor=black,linecolor=black](4.0,2.0){0.05}
\pscircle[fillstyle=solid,fillcolor=black,linecolor=black](4.0,3.0){0.05}
\pscircle[fillstyle=solid,fillcolor=black,linecolor=black](5.0,-1.0){0.05}
\pscircle[fillstyle=solid,fillcolor=black,linecolor=black](5.0,0.0){0.05}
\pscircle[fillstyle=solid,fillcolor=black,linecolor=black](5.0,1.0){0.05}
\pscircle[fillstyle=solid,fillcolor=black,linecolor=black](5.0,2.0){0.05}
\pscircle[fillstyle=solid,fillcolor=black,linecolor=black](5.0,3.0){0.05}
\psline[linewidth=0.6mm,linecolor=blue]{-}(1.0,3.0)(2.0,1.0)(4.0,0.0)(5.0,0.0)
\psline[linewidth=0.6mm,linecolor=red]{-}(0.0,3.0)(0.0,2.0)(1.0,2.0)(1.0,1.0)(1.0,0.0)(2.0,0.0)(3.0,0.0)(3.0,-1.0)(4.0,-1.0)(5.0,-1.0)
\psline[linewidth=0.6mm,linecolor=blue,linestyle=dotted]{-}(1.0,3.0)(2.0,3.0)(2.0,2.0)(2.0,1.0)(3.0,1.0)(4.0,1.0)(4.0,0.0)(5.0,0.0)
\psline[linewidth=0.6mm,linecolor=red,linestyle=dotted]{-}(0.0,3.0)(0.0,2.0)(1.0,0.0)(3.0,-1.0)(4.0,-1.0)(5.0,-1.0)
\end{pspicture}
}{The output approximations}
\end{tabular*}
\caption{Illustration of the grid search algorithm for systemic risk measures.  The blue grid points represent
those points that are known to be in $R$ and the red grid points represent those points that are known not to be inside.}\label{fig:algo}
\end{figure}

Set-valued systemic risk measures might, at first sight, seem too complicated to be computed easily, but, in fact, they are not. In this section, we construct a  simple grid search algorithm that provides an approximation to systemic risk measures. This method was implemented in the context of the numerical examples that are presented in Section \ref{Sec:casestudies}.

The suggested grid search algorithm approximates the boundary of the systemic risk measure on a given compact area with an accuracy that is specified a priori. The method is illustrated in Figure \ref{fig:algo}. The algorithm constructs an upper set $\bar R(Y;k)$ of $R(Y;k)$ that is very close to $R(Y;k)$ in the following sense.
\begin{defi}
Let $v \in \bbr^l_{++}$. A \emph{$v$-approximation} of the systemic risk measure $R$ is a function $\bar R: \ycal \times \bbr^l \to \pcal(\bbr^l;\bbr^l_+)$ such that 
$$ \bar R(Y;k) \subseteq R(Y;k) \subseteq \bar R(Y;k) - v \quad (\forall k \in \bbr^l, \; \forall Y \in \ycal). $$
\end{defi}
 
Due the cash-invariance of $R$, we assume that also the approximation $\bar R$ is cash-invariant and focus w.l.o.g.\ on $k=0$. 

We will now explain the algorithm as illustrated in Figure \ref{fig:algo}.   
First the compact box of interest is subdivided by grid points; the distance between grid points determines the accuracy of the algorithm. The algorithm seeks to determine the marked area in Figure \ref{fig:algo}(a). For this purpose, the algorithm will verify for all nodes $m \in \bbr^l$ on the grid, if the condition  $Y_{m} \in \acal$   is satisfied. This is equivalent to  $m \in R (Y)$.   For illustrative purposes, we will assume that each sub-box is a cube in $\bbr^l$.

Figure \ref{fig:algo}(b) shows the first step of the algorithm, a bisection search along a diagonal of the box with direction $(1,1, \dots, 1)^\T \in \bbr^l$. This leads to the identification of two nodes which are neighbors on this diagonal, one inside  $R(Y)$  and one outside. Since systemic risk measures are upper sets with ordering cone $\bbr^l_+$,  $m \in R(Y)$  implies that  $m  + \bbr^l_+ \subseteq R(Y)$,  while  $m\not\in R(Y)$   implies   $m  - \bbr^l_+ \subseteq R(Y)^c$,  the complement of  $R(Y)$.   This is indicated by the blue and red nodes in Figure \ref{fig:algo}(c).

Whenever two nodes are neighbors in the direction of the vector $(1,1, \dots, 1)^\T$, one belonging to  $R (Y)$ 
and the other not, the boundary of  $R (Y)$  intersects the small sub-box defined by these two points. The
algorithm proceeds with checking the other edges of the sub-box. The fact that  $R (Y)$  is an upper set with ordering cone $\bbr^l_+$ can again be exploited in each step as described above. This part of the procedure is illustrated in Figures \ref{fig:algo}(c)-(f).

The algorithm terminates when all points in the box of interest were either identified as elements of  $R (Y)$  or
its complement, see Figure \ref{fig:algo}(g). The boundary of  $R (Y)$  thus lies in between the red and blue surface, as indicated in Figure \ref{fig:algo}(h). Instead of the outer polyhedral approximation, its convex hull (dotted red line) might be used as an outer approximation if convexity is desired. Similarly, if convexity is not guaranteed, the inner convex approximation might be replaced by the inner polyhedral approximation (dotted blue line).

\begin{rem}
The suggested algorithm can be modified in order to exploit the advantages of parallel processing. Parallel computing is possible if successively finer grids are implemented. The same procedure can also be used in order to improve the precision of a previously derived approximation. 

The initial step (on a coarse grid) consists in the sequential algorithm that we proposed above.  In a second step, each sub-box that contains a part of the boundary of the systemic risk measure can then be recomputed with the sequential algorithm on a finer grid. The computations on each of these sub-boxes can be performed independently of each other. Depending on the number and speed of the available processors, a successively refined computation on a sequence of grids can be repeated multiple times.  

In order to improve the efficiency of parallel algorithms for systemic risk measures, the following ideas might be exploited in the future:
\begin{enumerate}
\item Checking if a node belongs to a systemic risk measurement or not, typically requires Monte Carlo simulations. The efficiency can thus be improved by implementing suitable variance reduction techniques. Moreover, the condition  $Y_{m} \in \acal$  does not require precise knowledge of   $Y_{m}$.  In particular, if the simulation results indicate that  $Y_{m}$  clearly lies inside or outside $\acal$, less accuracy is sufficient than in boundary cases. This intuition could be made precise by formulating appropriate stopping rules for simulation procedures that are inspired by results from sequential statistics, see e.g.\ \ci{Siegmund}.
\item Since the values of systemic risk measures are upper (convex) sets, image processing techniques might be applied to improve the final results of the computations and refine the approximations further. 
\end{enumerate}
\end{rem}

\begin{rem}
The algorithm can also be used for high dimensional problems -- at the price of substantially larger computation times and required memory capacity. The dimension can, however, be reduced if the entities can be subdivided into groups with equal capital requirements, as described in Example~\ref{network_ex}(iv). Such an approach is  realistic, if regulatory requirements depend on the type of the firms (see e.g.\ \ci{dodd-frank}). 
\end{rem}

\begin{rem}
If only  EARs  are computed, it is unnecessary to determine the whole boundary of the underlying systemic risk measurement. 
\begin{enumerate}
\item
According to  Lemma~\ref{thm:scalarization}, the aim is to compute a point in the boundary with normal vector $w: \ycal \to \bbr^l_{++} $ . A good strategy consists in applying the grid search algorithm sequentially on sub-boxes, but only focusing on those boxes that could contain a point on the frontier with supporting hyperplane with normal vector $w(Y)$. 
\item Utilizing Lemma~\ref{thm:scalarization}, observe that the computation of an EAR is equivalent to finding the solution of a convex optimization problem.  This allows the application of numerical techniques such as a pattern search or particle swarm methods. 
\end{enumerate}
\end{rem}

\section{Numerical case studies}\label{Sec:casestudies}

We illustrate the proposed systemic risk measures in case studies. Section~\ref{Sec:cs-agg} studies  different aggregation functions, see Examples~\ref{network_ex}(i) and (ii). The special case without sensitivity to capital levels corresponds to the setting of  \ci*{chen2013axiomatic} and \ci*{kromer2013systemic}.  Section~\ref{Sec:cs-network} investigates network models that are based on extensions of the basic network models of \ci{EN01}, cf.\ e.g.\ \ci{CFS05}, \ci{AFM15}, \ci{RV11}, \ci{AW_15}, \ci{fei15}, \ci{hurd15}. As mentioned in Example~\ref{network_ex}(iii), we include local interaction via direct credit links and global interaction induced by price impact. The dimension of the problem is reduced by considering a small number of homogeneous groups of financial institutions, see Example~\ref{network_ex}(iv).

\subsection{Comparative statics for aggregation functions}\label{Sec:cs-agg}
\begin{figure}[p]
\centering
\begin{minipage}{\textwidth}
\centering
\includegraphics[height=0.3\textheight,width=0.6\textwidth]{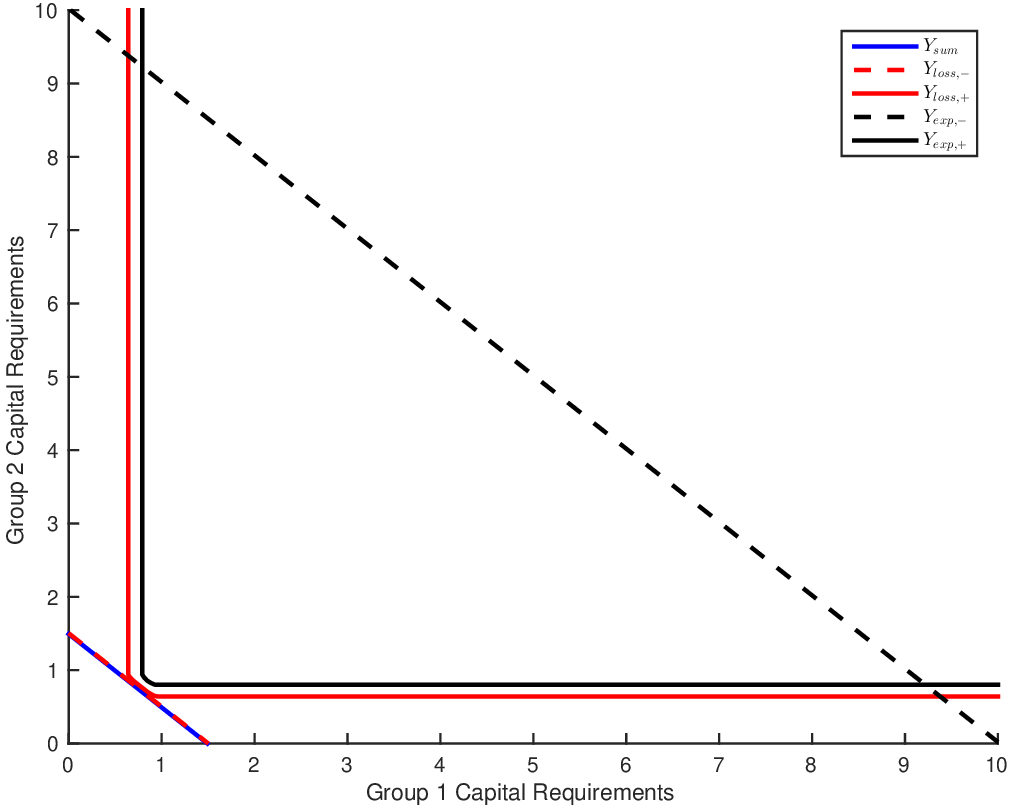}
\caption{Comparison of capital requirements for different aggregation functions}
\label{fig:agg-full}
\end{minipage}
\begin{minipage}{\textwidth}
~\\[1cm]
\end{minipage}
\begin{minipage}{\textwidth}
\centering
\includegraphics[height=0.3\textheight,width=0.6\textwidth]{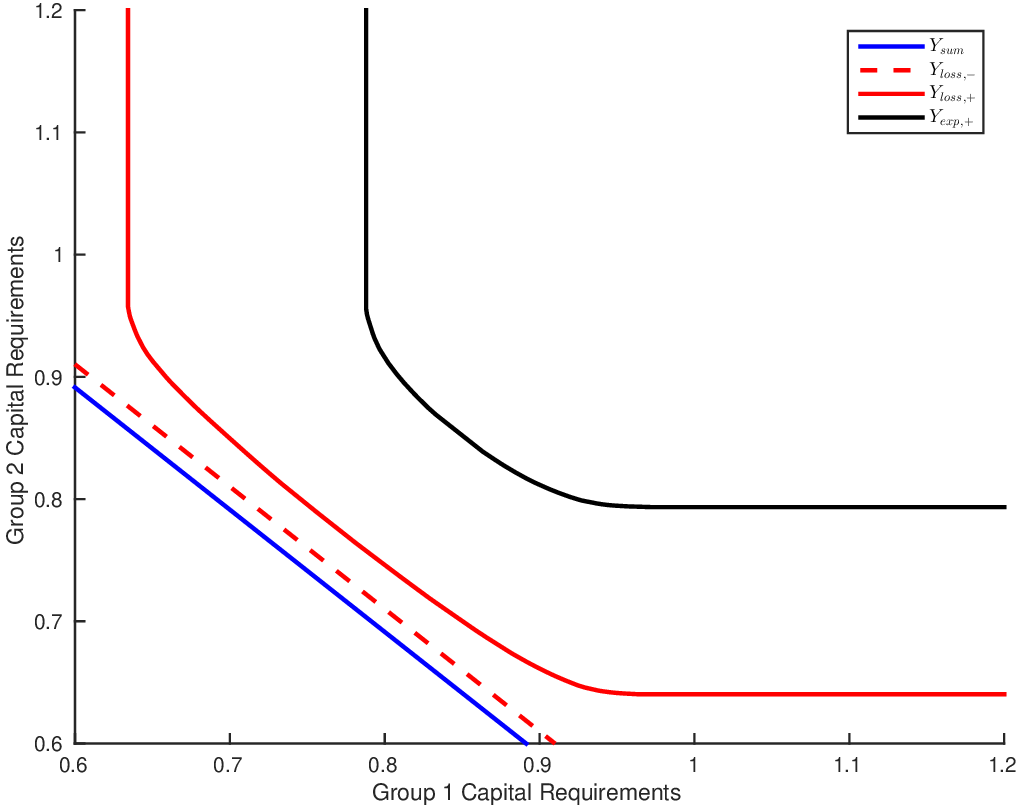}
\caption{Comparison of capital requirements for different aggregation functions on the interval $[0.6,1.2]^2$}
\label{fig:agg-zoom}
\end{minipage}
\end{figure}

Let us now investigate systemic risk measures with different aggregation functions in Examples~\ref{network_ex} (i) and (ii), cf.\ \ci*{chen2013axiomatic} and \ci*{kromer2013systemic}.

\begin{enumerate}
\item {\bf System-wide profit and loss:} 
 A simple example is  the aggregation function\footnote{Example 1 in \ci*{chen2013axiomatic}; Example 3.4 in \ci*{kromer2013systemic}}  \[\Lambda_{sum}(x) := \sum_{i = 1}^n x_i\] where $x_i$ is interpreted as the profit of financial firm $i$; a negative profit does, of course, indicate a loss.  Profits are simply aggregated additively across entities. The system is thus modeled like a portfolio which is diversified across different assets.  For this choice of $\Lambda_{sum}$ both aggregation mechanisms in Examples~\ref{network_ex} (i) and (ii), either insensitive  or sensitive to capital levels, are identical.   
 
\item {\bf System-wide losses:}  Another aggregation function\footnote{Example 2 in \ci*{chen2013axiomatic}; Example 3.5 in \ci*{kromer2013systemic}} 
focuses only on the downside, but not on the upside risk:  \[\Lambda_{loss}(x) := \sum_{i = 1}^n -x_i^-\] where $z^- := \max(-z,0)$. Unlike $\Lambda_{sum}$ in the previous example, this aggregation function does not allow for a profitable firm to subsidize the losses of a distressed firm. For $\Lambda_{loss}$ the aggregation mechanisms in Examples~\ref{network_ex} (i) and (ii) are different.

\item {\bf Exponential loss:}  The aggregation function\footnote{Example 3.6 in \ci*{kromer2013systemic}}  \[\Lambda_{exp}(x) := \sum_{i = 1}^n [1 - \exp(2x_i^-)]\] focuses on the downside risk only, but aggregates losses non-linearly. It could be used if a regulatory authority prefers multiple small losses over one large loss.   It would, however, provide regulatory incentives for entities to split into multiple sub-entities.  For $\Lambda_{exp}$ the aggregation mechanisms in Examples~\ref{network_ex} (i) and (ii) are different.
\end{enumerate}

In numerical case studies we will now illustrate how capital requirements depend on the choice of the aggregation function $\Lambda_{sum}$, $\Lambda_{loss}$, and $\Lambda_{exp}$. For simplicity, we assume for $n=100$ that $X=(X_i)_{i=1,\dots, 100}$ is a random vector with lognormal margins and Gaussian copula, i.e.
$$X_i \sim \exp(\mu + \sigma N_i) + b, \quad\quad i=1, \dots, 100,$$ 
where $N=(N_i)_{i=1, \dots, 100}$ is jointly Gaussian with standardized margins and a common pairwise correlation of $80\%$.  We fix $\sigma = 1$, and choose $\mu$ and $b$ such that the maximum loss for each entity is  $-1$ with a marginal probability of $25\%$ for incurring a loss; this corresponds to $b=-1$ and $\mu = \Phi^{-1}(0.75) \approx 0.6745$ where $\Phi$ is the cumulative distribution function of the standard normal. The Gaussian copula should not be considered a realistic choice. It provides a model for the underlying factors in which tail dependence is not present. As we will see, systemic interaction still creates overall risk. 

The 100 financial firms are divided into two symmetric groups of 50 firms each. We assume that the capital levels of the firms in a group are equal, i.e.\ we choose, analogous to  Example~\ref{network_ex} (iv), the function 
$$g(k_1, k_2) = (\underbrace{k_1, k_1, \dots, k_1}_{50},  \underbrace{k_2, k_2, \dots, k_2}_{50} )^\T.$$ 

We can now specify five  CVMs  that capture the random outcomes for the three aggregation functions and two aggregation mechanisms, either sensitive (labeled by +) or insensitive (labeled by -) to capital levels $k= \left( \begin{array}{c} k_1 \\ k_ 2 \end{array} \right)$:
\begin{eqnarray*}
Y_{sum} (k)      & \; =  \;&   \Lambda_{sum} (g(k)+X) = \Lambda_{sum} (X) +  50 \cdot (k_1 + k_2)   \\
Y_{loss,-} (k)     &  \; =  \; &  \Lambda_{loss} (X) +  50 \cdot (k_1 + k_2)   \\
Y_{loss,+} (k)    &  \; =  \; &  \Lambda_{loss} (g(k)+X) \\
Y_{exp,-}  (k)    &  \; =  \; &  \Lambda_{exp} (X) +  50 \cdot (k_1 + k_2)   \\
Y_{exp, + } (k)  &  \; =  \; &  \Lambda_{exp} (g(k)+X) 
\end{eqnarray*}

As the last ingredient of our numerical case studies we need to specify an acceptance criterion in terms of a scalar monetary risk measure. The qualitative results that we discuss below are very similar for any of the choices described in Example~\ref{ex:acceptance}(i). Let us e.g.\ take a shifted (and thus less conservative)  optimized certainty equivalent. Defining $\rho_{\rm{OCE}}$ according to  Example~\ref{ex:acceptance}(i)(c) with $u(x) = \log(x+1)$, we set for suitably integrable random variables $M$:
$$\rho (M) = \rho_{\rm{OCE}} (M) - 10 ,$$
and define acceptability in terms of the acceptance set of the risk measure $\rho$, i.e.
\[\acal = \{M  \; | \; \exists \eta \in \bbr: \eta + E[\log(M - \eta + 1)] \geq -10\}.\]
Figure \ref{fig:agg-full} plots the capital requirements for each firm in group 1 against the capital requirements for each firm in group 2.  Due to the symmetry of the groups the capital allocations are symmetric about the $45^{\circ}$ line. The graph demonstrates several important points.  

First, if the aggregation mechanisms are insensitive to capital levels (-), the minimal capital allocations on the boundary of the systemic risk measurement form a hyperplane -- in two dimensions: a line -- due to the cash-invariance of the scalar risk measure $\rho$. Second, the exponential aggregation insensitive to capital levels ($Y_{exp, -}$) is generally very conservative. In contrast, exponential aggregation sensitive to capital levels ($Y_{exp, +}$) allows for a significant reduction in capital. This effect is reversed if system-wide losses are considered. In this case, aggregation that is sensitive to capital ($Y_{loss, +}$) is more conservative than if insensitive to capital ($Y_{loss, -}$), as can also be observed in Figure \ref{fig:agg-zoom} that plots the part $[0.6,1.2]^2$ and zooms into Figure \ref{fig:agg-full}. This can be understood as follows. For system-wide losses, capital provides essentially a linear relief. However, if capital is added prior to aggregation, a fraction of it might be neglected due to the cut-off at zero. Formally, $Y_{loss,+} \leq Y_{loss, - }$.  This implies that aggregation sensitive to capital levels leads to more  conservative  systemic risk measures than aggregation that is insensitive to it. The exponential aggregation function aggregates  losses instead superlinearly, implying that adding capital prior to aggregation ($Y_{exp, +}$)  provides \emph{more} relief than adding it afterwards ($Y_{exp, -}$), as long as the capital level of each entity is not too low.  

A third point is that for both aggregation mechanisms $Y_{loss,+}$ and $Y_{exp, +}$ the reduction of the capital levels of one group below a certain level cannot be offset by increasing the capital for the other group anymore. For both aggregation functions profits are cut off in the aggregation mechanisms which prevents upside risk from having an impact on the random outcome. This limits the potential for tradeoffs.

\begin{figure}
\centering
\includegraphics[height=0.3\textheight,width=0.6\textwidth]{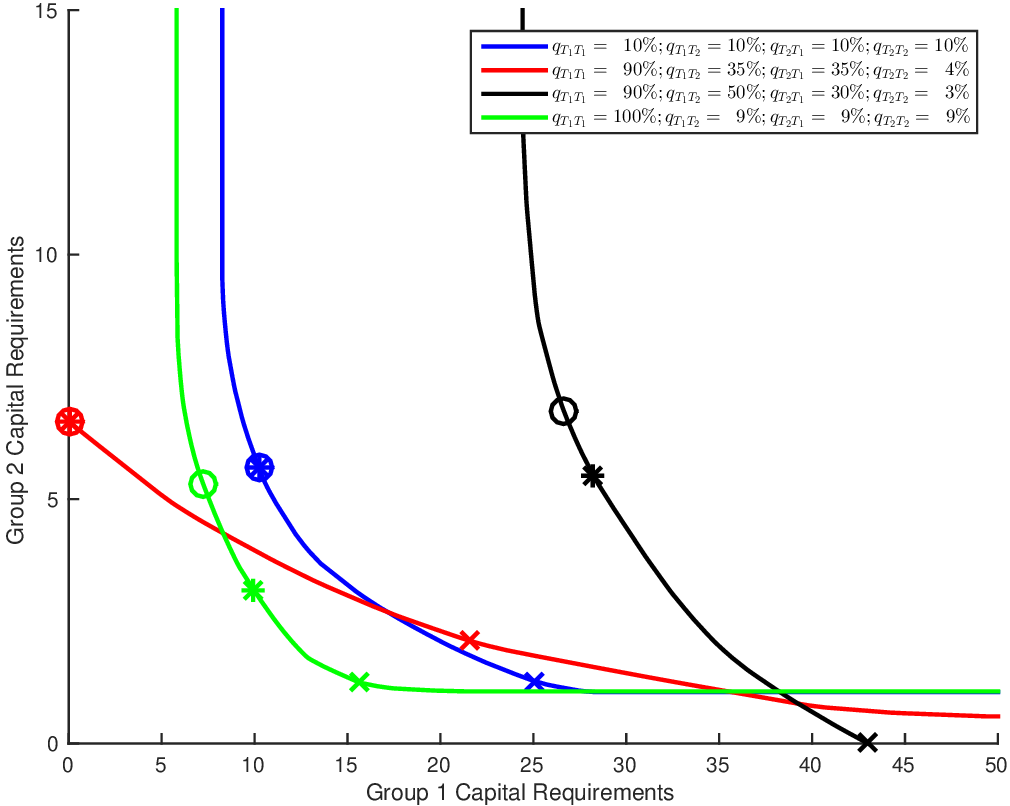}
\caption{Comparison of capital requirements and  EARs  for network scenarios A1-A4 in Example \ref{ex:2tier}}
\label{fig:eisenberg}
\end{figure}

\subsection{Comparative statics for network models}\label{Sec:cs-network}

Another example for which systemic risk measures can be computed are financial network models as originally suggested by \ci{EN01}. As described in Example~\ref{network_ex}(iii), we consider a system of $n+1$ agents
$$ \{0, 1,2,\dots,n\}.$$
The entities $1, 2, \dots, n$ are financial institutions, entity $0$ is the society. The financial institutions are endowed with capital 
$$k =(k_1, k_2, \dots, k_n )^\T\in \bbr^n,$$
and a random number of liquid assets $X= (X_1, X_2, \dots, X_n)^\T $ and illiquid assets $S= (S_1, S_2, \dots, S_n)^\T$. Their value is revealed at time $t=1$ in the one-period economy. Following  \ci*{CFS05}, we assume that agents are connected in a network of bilateral debt agreements that are cleared immediately after time $t=1$. The resulting payments are transfers of the liquid asset. If an agent does not hold a sufficient amount of the liquid asset, she must sell a fraction of the holdings in the illiquid asset. This triggers price impact. A mathematical analysis of a comprehensive model which includes the current example as a special case can be found in \ci{AW_15}.
 
Let us briefly explain the mechanics of the clearing mechanism. An agent may be creditor or obligor to other agents in the network. We denote the nominal liability of
agent $i$ to agent $j$ by $\bar p_{ij} \geq 0$. Setting $\bar p_{ii}=0$,
we define the liability matrix $\bar \Pi = (\bar p_{ij})_{i,j= 0, 1, \dots, n}$.  As a consequence, the total nominal liabilities of agent $i$ are
$$\bar p_i = \sum_{j=0}^n \bar p_{i j} .$$
The total liabilities of agents in the economy are thus encoded by $\bar p= (\bar p_i) \in \bbr^{n+1}_+$.
The relative liabilities are defined as
$$ a_{ij} = \left\{
\begin{array}{ll}
\frac {\bar p_{ij}}{\bar p_i}, & \bar p_i>0,\\
0, & \bar p_i =0.
\end{array}
\right. $$
We set $A= (a_{ij})_{i,j = 0, 1, 2, \dots, n}$ with  $\sum_{j} a_{ij} = 1$ if $\bar p_i > 0$.
The agents $i = 1,2, \dots, n$ may default on their promised payments $\bar p_i$, if sufficient funds are not available; the society, i.e.\ agent $0$, will never default on any of its obligations.

In the event of default, obligations are partially fulfilled. If none of the agents holds the illiquid asset, the clearing
mechanism proceeds as described in \ci{EN01}. If illiquid assets are held and if liquid assets are insufficient to
cover the financial obligations, illiquid assets will be sold in order to obtain the necessary funds. This will trigger a decrease of the price of the illiquid asset, that is captured by an inverse demand function $f: \bbr_+ \to \bbr_{++}$, mapping quantity into price per share. The following assumption guarantees the uniqueness of an equilibrium\footnote{If the equilibrium is not unique, one could instead analyze the largest (or smallest) equilibrium. In the context of the current paper, uniqueness is not a key issue. If endowments are strictly positive, uniqueness can however be proven under Assumption~\ref{ass:idf}. This was shown in \ci*{AFM15unique}.}. 
\begin{assumption}\label{ass:idf}
The inverse demand function $f: \bbr_+ \to \bbr_{++}$ is nonincreasing and $y \in \bbr_+ \mapsto y f(y)$ is a strictly increasing function.
\end{assumption}

The adjusted clearing mechanism is motivated in \ci*{CFS05}. We refer to these papers for a more detailed
discussion of its economic significance.  Let $x,s \in \bbr^n_+$ denote the holdings in the liquid and illiquid asset of the financial institutions, respectively. We compute the \emph{realized payment} or \emph{clearing vector}
$p(x;s)\in \bbr^{n+1}_+$ and implied \emph{clearing price} $\pi(x;s) \in \bbr_+$ of the illiquid asset  as the solution to the following fixed point problem:
\begin{align}
\label{clearing} p_i(x;s)  & = \bar p_i   \;\; \wedge \; \;   \left(\sum_{j=0}^n p_j(x;s) a_{ji} + x_i + \pi(x;s) s_i\right) , \quad\quad i =1,2, \dots, n\\
\label{price} \pi(x;s) & = f\left[\sum_{i=1}^{n} \left(\frac{1}{\pi(x;s)} \left[\bar{p}_i-x_i-\sum_{j=0}^{n} a_{ji}p_j(x;s)\right]^{+}  \;\; \wedge \;\; s_i  \right) \right]
\end{align}
and $p_0(x) = \bar{p}_0$. We denote the unique fixed point of \eqref{clearing} and \eqref{price} by $(p(x;s),\pi(x;s))$. The special case of  \ci{EN01} consists in allocating none of the endowment to the illiquid asset, i.e.\ $s = 0 \in \bbr^n_+$.

For  any agent $i= 0, 1, 2, \dots, n$, the sum of debt and equity is  given by
$$\sum _{j\neq i} p_j(x;s) a_{ji} +x_i +\pi(x;s)s_i ,$$
where we set $x_0 = \bar p_0$ and $s_0 = 0$. If we subtract the promised obligations, we obtain
\begin{equation}\label{eq:el}
 e_i(x;s) = \sum_{j\neq i} p_j(x;s) a_{ji} +x_i +\pi(x;s)s_i - \bar p_i  , \quad i=0,1,2,\dots, n.
\end{equation}

\begin{lemma}
\label{lemmae}
The function $e_0$ defined in equation \eqref{eq:el} is non-decreasing in the first component, continuous, and uniformly bounded from above and below on $\bbr^n_+ \times \bbr^n_+$.
Moreover, $x \mapsto e_0(x;0)$ is concave.
\end{lemma}

In the numerical case studies, we construct the bilateral obligations of the firms via random networks. The probability of a connection from firm $i$ to firm $j$ is denoted by $q_{ij} \in [0,1]$.  We set $q_{i0} = 1$ and $q_{0i} = 0$ for all $i$. The connections are independently sampled.  The size of an obligation attached to a sampled directed link from firm $i$ to firm $j$ is $w_{ij}$, i.e.\ the actual obligation of
$i$ to $j$ equals $w_{ij}$, if there is a link, and $0$ otherwise:  
$$\bar{p}_{ij} = \begin{cases}w_{ij} &\text{with probability
} q_{ij}\\ 0 &\text{with probability } 1-q_{ij}.\end{cases}$$  
Since a firm does not have a liability to itself, we assume that $q_{ii}
= 0$ for all $i$. 

For simplicity, we construct all graphs in a grouped structure. Letting
$(T_j)_{j=1}^l$ be a partition of the set $\{1,2, \dots,n \} $ into $l$ groups,  we assume that the probabilities of connections and the weights depend only on the groups to which firms belong to; the group that contains the society is $T_0 = \{0\}$. With a slight abuse of notion, we write
$$ 
\begin{array}{ccc}
q_{ij}  &= &
q_{T_{(i)}T_{(j)}}
 \\
 w_{ij} &  =  &w_{T_{(i)}T_{(j)}}
\end{array}
\quad\quad\quad\quad\quad
\left(
i \in T_{(i)}, j \in T_{(j)}
\right)
$$ 
where for  firm $i$ its group  is labeled by $(i) \in \{1, 2, \dots, l \} $. 
The special case $l=n$ is a general random graph structure, the case  $l = 1$ is (apart from node $0$) an Erd\H{o}s-R\'{e}nyi random graph.

We assume that capital requirements are equal for all financial firms within the same group. 

\begin{table}
\centering
\begin{tabular}{|ll|c|c|c|c|}
\hline
&& $q_{T_1T_1}$ & $q_{T_1T_2}$ & $q_{T_2T_1}$ & $q_{T_2T_2}$\\
\hline
\multirow{4}{*}{\parbox{3.1cm}{Varying Network Structures}} & Scenario A1 & $10\%$ & $10\%$ & $10\%$ & $10\%$\\
& Scenario A2 & $90\%$ & $35\%$ & $35\%$ & $4\%$\\
& Scenario A3 & $90\%$ & $50\%$ & $30\%$ & $3\%$\\
& Scenario A4 & $100\%$ & $9\%$ & $9\%$ & $9\%$\\
\hline 
\multirow{4}{*}{\parbox{3.1cm}{Fixed Inter-Group\\ Connections}} & Scenario B1 & $60\%$ & $20\%$ & $20\%$ & $30\%$\\
& Scenario B2 & $60\%$ & $20\%$ & $20\%$ & $10\%$\\
& Scenario B3 & $10\%$ & $20\%$ & $20\%$ & $30\%$\\
& Scenario B4 & $10\%$ & $20\%$ & $20\%$ & $10\%$\\ 
\hline
\multirow{6}{*}{\parbox{3.1cm}{Fixed Intra-Group\\ Connections}} & Scenario C1 & $60\%$ & $20\%$ & $20\%$ & $30\%$\\
& Scenario C2 & $60\%$ & $50\%$ & $50\%$ & $30\%$\\
& Scenario C3 & $60\%$ & $30\%$ & $10\%$ & $30\%$\\
& Scenario C4 & $60\%$ & $60\%$ & $40\%$ & $30\%$\\
& Scenario C5 & $60\%$ & $10\%$ & $30\%$ & $30\%$\\
& Scenario C6 & $60\%$ & $40\%$ & $60\%$ & $30\%$\\
\hline
\end{tabular}
\caption{Summary of the comparative scenarios for example~\ref{ex:2tier}}
\label{table:2tier-scenarios}
\end{table}

\begin{figure}
\centering
\includegraphics[height=0.3\textheight,width=0.6\textwidth]{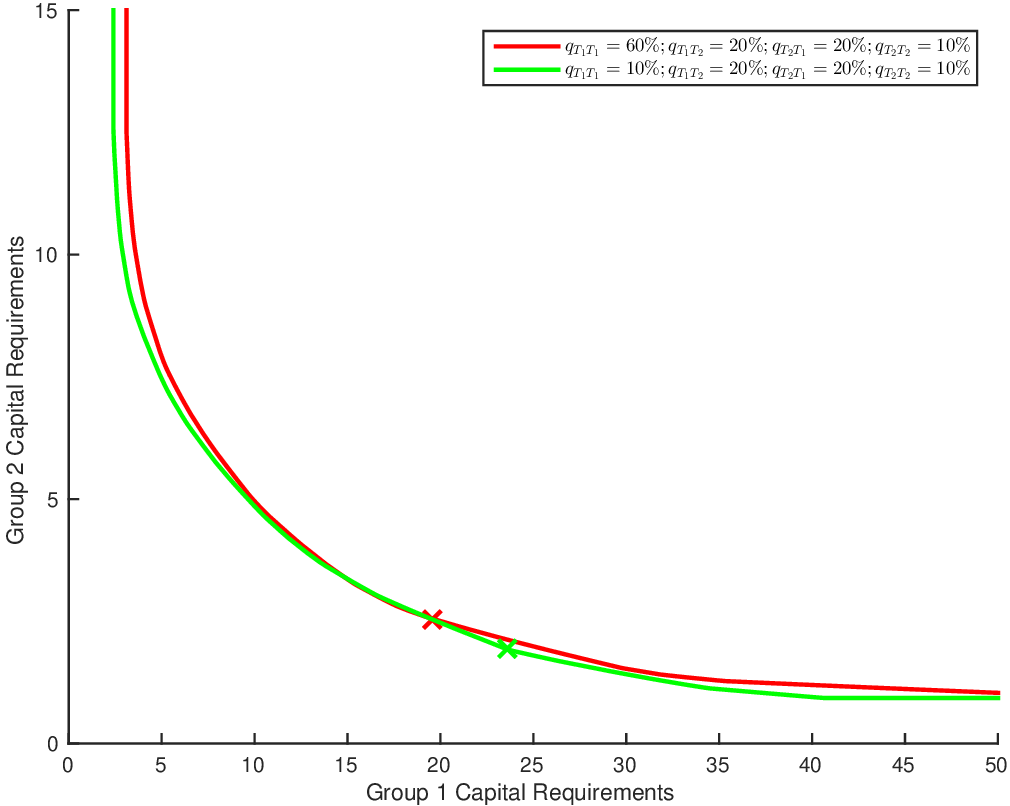}
\caption{Comparison of capital requirements and  EARs  for network scenarios B2 and B4 in Example \ref{ex:2tier}}
\label{fig:eisenberg-changeT1}
\end{figure}

\begin{figure}
\centering
\includegraphics[height=0.3\textheight,width=0.6\textwidth]{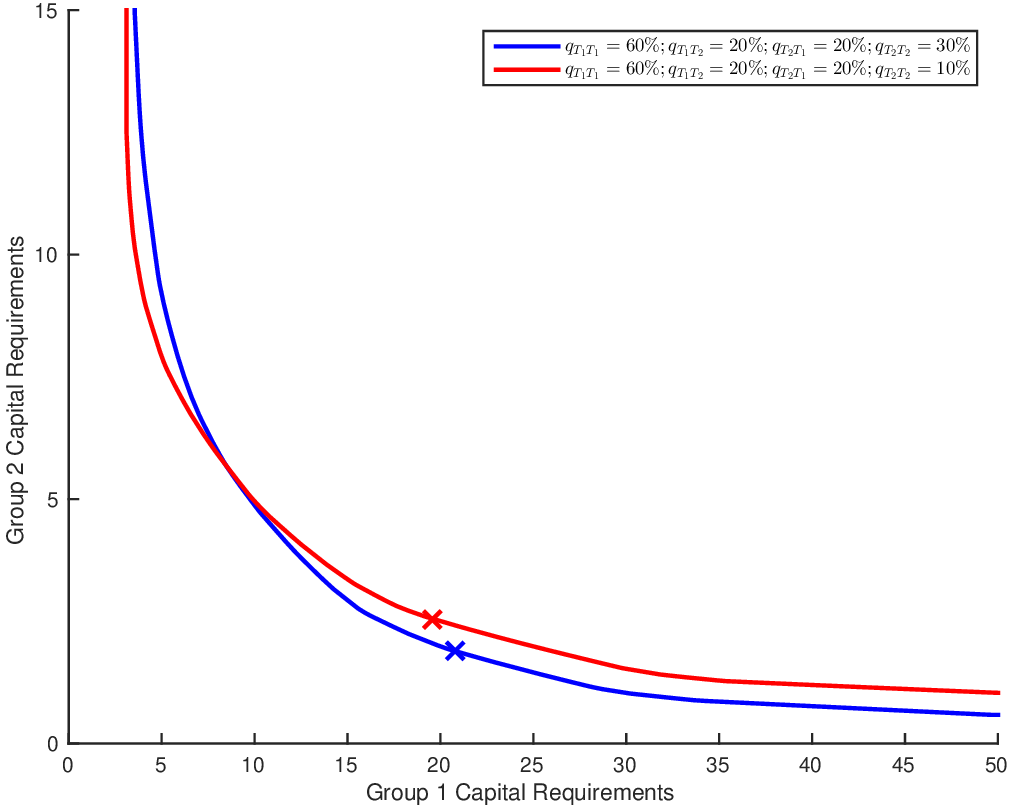}
\caption{Comparison of capital requirements and EARs for network structures B1 and B2 in Example \ref{ex:2tier}}
\label{fig:eisenberg-changeT2}
\end{figure}
\begin{figure}
\centering
\includegraphics[height=0.3\textheight,width=0.6\textwidth]{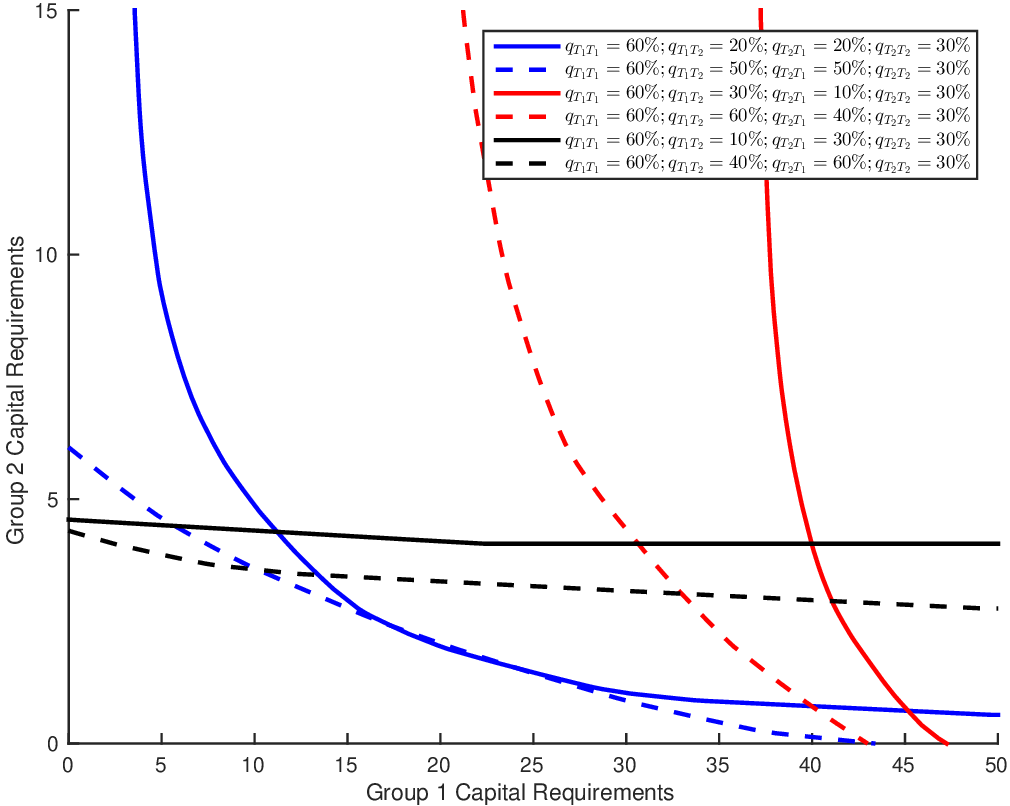}
\caption{Comparison of capital requirements for network structures C1-C6 in Example \ref{ex:2tier}}
\label{fig:eisenberg-changeInter}
\end{figure}
\begin{example}\label{ex:2tier}

As a first example, we consider a network of obligations with two groups of firms. We assume that the firms do not own any illiquid asset. The resulting clearing mechanism is thus exactly as described in  \ci{EN01}. 

Let the total number of firms be $n=100$. We suppose that $T_1$ contains $10$ large firms and $T_2$ the remaining 90 small firms. We compare different probabilities of links between firms which can be interpreted as varying the density of the network of obligations within and between groups. These probabilities are given in Table~\ref{table:2tier-scenarios}.

The following ingredients are fixed in all case studies:
\begin{enumerate}
\item If there are links, large firms owe each other $10$, small firms $1$. Connections between large and small firms are of size $2$.  Additionally, each large firm owes $10$ and each small firms owes $1$ to the society.
\item  The random vector $X = (X_i)_{i=1,2,\dots, 100 }$ of the liquid assets has one-dimensional margins distributed as ${Beta}(2,5)$ and a Gaussian copula with a correlation of $0.5$ between any two firms.
\item  The acceptance criterion is given in terms of a shifted (more conservative and not normalized) average value at risk at level $\lambda = 1\%$. The random payments to the society are acceptable if their ${\rm AV@R}_{\lambda} $ plus $90\%$ of the  payments promised to the society is at most $0$.
\item In order to avoid technical complications in the case of network models, we only consider capital allocations in $\bbr_+^l$ and adjust equation~\eqref{eq:def-risk} accordingly. The construction of the systemic risk measurements and corresponding  EARs  are altered consistently with this constraint. 
\end{enumerate}

\begin{table}
\centering
\begin{tabular}{|cr|c|c|c|}
\hline
&& \multicolumn{3}{|c|}{$w(Y)$} \\ \cline{3-5}
&& $(1,1)$ & $(10,90)$ & $(\frac{1}{10\bar p_{T_1}} , \frac{1}{90\bar p_{T_2})}$ \\ \hline
Scenario A1 & $\begin{array}{r}\text{Group 1} \\ \text{Group 2}\end{array}$ & $\begin{array}{r@{.}l}10&2102 \\ 5&6607\end{array}$ & $\begin{array}{r@{.}l}25&0751 \\ 1&2763\end{array}$ & $\begin{array}{r@{.}l}10&2102 \\ 5&6607\end{array}$ \\  \hline
Scenario A2 & $\begin{array}{r}\text{Group 1} \\ \text{Group 2}\end{array}$ & $\begin{array}{r@{.}l}\;\,0&0000 \\ 6&5831\end{array}$ & $\begin{array}{r@{.}l}21&5716 \\ 2&1021\end{array}$ & $\begin{array}{r@{.}l}\;\,0&0000 \\ 6&5831\end{array}$ \\ \hline
Scenario A3 & $\begin{array}{r}\text{Group 1} \\ \text{Group 2}\end{array}$ & $\begin{array}{r@{.}l}26&6266 \\ 6&8168\end{array}$ & $\begin{array}{r@{.}l}43&0180 \\ 0&0000\end{array}$ & $\begin{array}{r@{.}l}28&2282 \\ 5&4805\end{array}$ \\ \hline
Scenario A4 & $\begin{array}{r}\text{Group 1} \\ \text{Group 2}\end{array}$ & $\begin{array}{r@{.}l}\;\,7&2573 \\ 5&2853\end{array}$ & $\begin{array}{r@{.}l}15&6657 \\ 1&2462\end{array}$ & $\begin{array}{r@{.}l}\;\,9&9600 \\ 3&1381\end{array}$ \\ \hline
\end{tabular}
\caption{Value of  EARs  for network scenarios A1-A4 in Example \ref{ex:2tier}}
\label{table:eisenberg}
\end{table}

Figure~\ref{fig:eisenberg} displays the capital requirements for both groups, presented by the (inner approximation) of the efficient frontier of the set under network scenarios A1-A4 described in Table~\ref{table:2tier-scenarios}.  It also shows the value of three different  EARs  which are  listed in Table~\ref{table:eisenberg}. 
As expected, the figure demonstrates that systemic risk measurements  are very sensitive to the levels of interconnectedness between the groups of banks and indicate that this relationship is complex.

In order to better understand the underlying relationships we look at 10 further network structure described in Table~\ref{table:2tier-scenarios}. We first focus on the situation in which the network structure between the two groups is fixed and vary the network connectivity within the groups. In Table~\ref{table:2tier-scenarios} the corresponding parameters are labelled ``Fixed Inter-Group Connections''. 

Figure~\ref{fig:eisenberg-changeT1} displays what happens if the connectivity is changed only within the group of large firms. Capital requirements increase if the connectivity within the group of large firms is increased, although the fundamental data of the economy are not modified. This observation does not depend on the capital levels of the small firms. We conducted several case studies with other parameters than those in the current case study and found that this effect is very robust within the current model setup. Despite the fact that increasing the connectivity within the group of large firms leads to uniformly stricter capital requirements, the  EAR  that minimizes the total regulatory capital offers capital relief to the responsible large firms at the expense of larger capital requirements of the small firms. This indicates that scalar systemic risk measures -- in contrast to set-valued systemic risk measures -- can actually be misleading.

If the connectivity within the group of large firms is fixed, but the connectivity of the small firms is varied, the direction of the effect depends on the size of the capital endowment of the large firms, as shown in  Figure~\ref{fig:eisenberg-changeT2}. Also this observation seems to be robust within the current model setting. If large firms hold a large amount of capital, increasing the connectivity between the small firms decreases the capital requirements. This can be interpreted as spreading solvency throughout the network. If, however, large firms hold only small amounts of capital, increasing the connectivity between the small firms increases capital requirements. Systemic risk is higher indicating that the network contributes to stronger default cascades in this case. The  EARs  that minimize total capital are in the example in the region where large firms hold a large amount of capital. If the small firms increase their connectivity, the small firms benefit from this  EAR  at the expense of the large firms by more than the induced capital reduction that would occur if the capital of the large firms was held constant.  

Finally, we investigate the impact of changing the connectivity between the firms while leaving the connectivity within the groups constant. In Table~\ref{table:2tier-scenarios} the corresponding parameters are labelled ``Fixed Intra-Group Connections''. The systemic risk measures for these six scenarios are displayed in Figure~\ref{fig:eisenberg-changeInter}. Scenarios 1 and 2 correspond to no net payments between the groups. Nevertheless, the obligations have an impact on systemic risk. Increasing the connectivity between the groups reduces systemic risk. In scenarios 3 and 4 the large firms are net borrowers; the total net amount is equal for both scenarios. Increasing connectivity between the groups again reduces systemic risk. In scenarios 5 and 6 the small firms are net borrowers, and again increasing the connectivity between the groups is beneficial in the current model framework.

\end{example}

\begin{example}\label{ex:3tier}

\begin{figure}[p]
\centering
\begin{minipage}{\textwidth}
\centering
\includegraphics[height=0.3\textheight,width=0.6\textwidth]{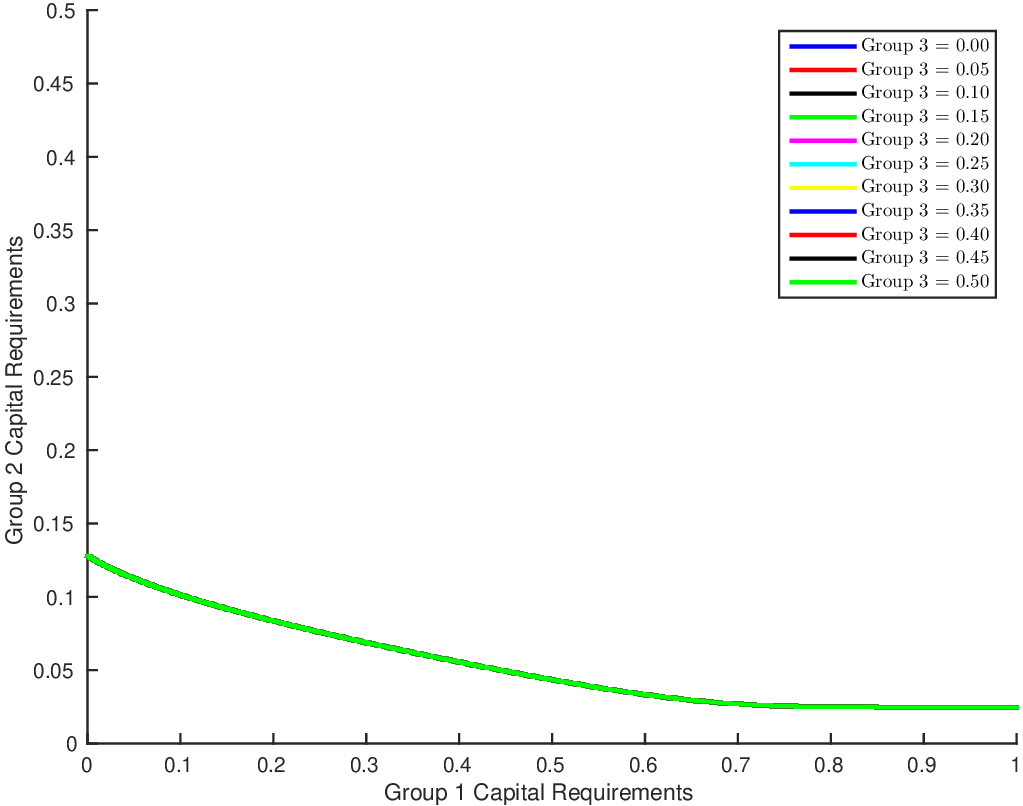}
\caption{Comparison of capital requirements for all three groups of firms (contours on levels for group 3) in Example~\ref{ex:3tier}}
\label{fig:3tier-cifuentes-3d}
\end{minipage}
\begin{minipage}{\textwidth}
~\\[1cm]
\end{minipage}
\begin{minipage}{\textwidth}
\centering
\includegraphics[height=0.3\textheight,width=0.6\textwidth]{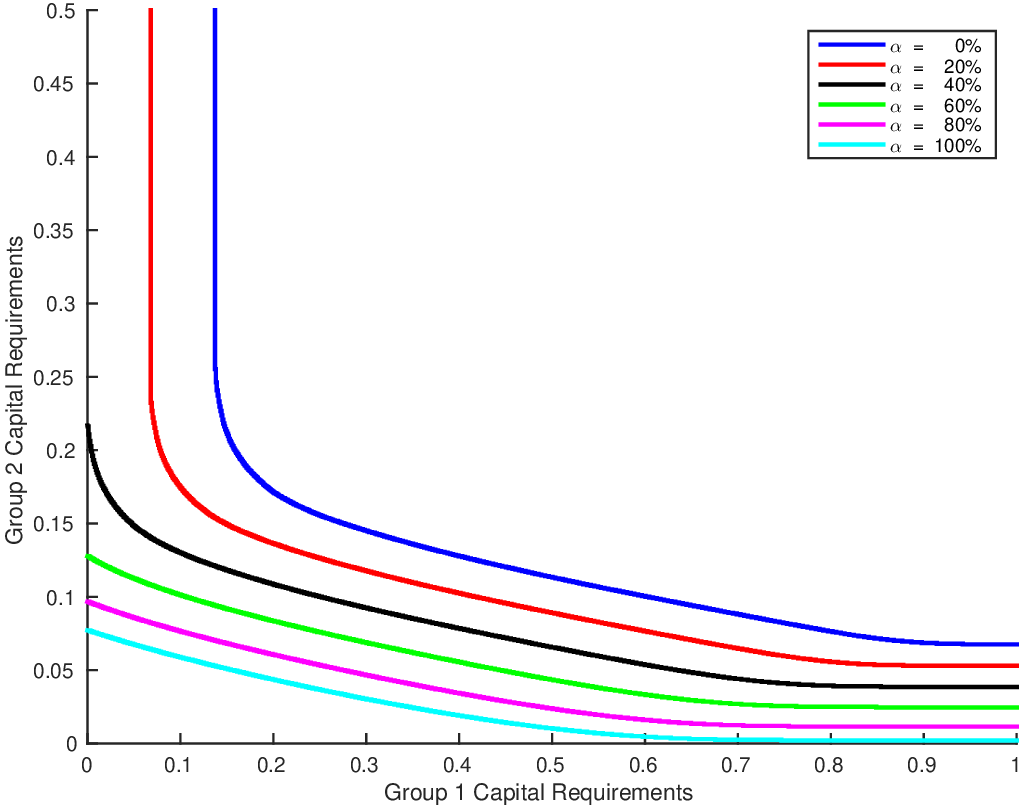}
\caption{Comparison of capital requirements for different fractions of net worth in the liquid asset in Example \ref{ex:3tier}}
\label{fig:3tier-cifuentes}
\end{minipage}
\end{figure}

The second case study considers a tiered graph with 3 groups of firms in a network model with both local and global interactions as described in \ci*{CFS05}, see also \ci{AW_15}.  We vary the fraction of endowment that is held in the illiquid asset, and keep all other parameters constant. We assume that in total the financial system consists of 300 firms, 10 in group $T_1$, 90 in group $T_2$ and 200 in group $T_3$. 
\begin{enumerate}
\item The probabilities for connections in the network of obligations and the weights assigned to them are fixed in this case study, see Table~\ref{table:3tier-network} for details. 
\begin{table}
\centering
\begin{tabular}{|cc|ccc|c|}
\hline
&& \multicolumn{4}{|c|}{To}\\
& $(q,w)$ & $T_1$ & $T_2$ & $T_3$ & $0$\\ \hline
\multirow{3}{*}{\rotatebox[origin=c]{90}{From}} & $T_1$ & $(100\%,\frac{10}{480})$ & $(80\%,\frac{3}{480})$ & $(60\%,\frac{3}{480})$ & $(100\%,\frac{10}{480})$ \\
& $T_2$ & $(60\%,\frac{3}{480})$ & $(40\%,\frac{2}{480})$ & $(20\%,\frac{2}{480})$ & $(100\%,\frac{2}{480})$\\
& $T_3$ & $(20\%,\frac{1}{480})$ & $(5\%,\frac{1}{480})$ & $(0\%,\frac{1}{480})$ & $(100\%,\frac{1}{480})$\\
\hline
\end{tabular}
\caption{Summary of the network topology for example~\ref{ex:3tier}}
\label{table:3tier-network}
\end{table}
\item The inverse demand function of the illiquid asset is given by
\[
f(x) := \begin{cases} 1 - \frac{2}{3}x & \text{if } x \leq \frac{1}{2} \\ \frac{\sqrt{2}}{3\sqrt{x}} & \text{if } x \geq \frac{1}{2} \end{cases}.
\]
Observe that this function satisfies Assumption~\ref{ass:idf} that guarantees uniqueness.
\item The random vectors $X = (X_i)_{i=1,2,\dots, n }$  and $S = (S_i)_{i=1,2,\dots, n }$ specifying the amount of the liquid and illiquid asset that is held by the agents is given by 
$$ X = \alpha \cdot N, \quad S = (1-\alpha) \cdot N  $$
for some random vector $N = (N_i)_{i=1,2,\dots, n}$ with 
$$\alpha \in \{0\%,20\%,40\%,60\%,80\%,100\%\}. $$
Observe that, since $f(0) =1$, the initial mark-to-market value of these assets is always given by the vector $N$, for any $\alpha$.

The distribution of $N$ is as follows. Dependence is specified via a Gaussian copula with a correlation of $0.5$ between any two firms. The marginal distributions depend on the groups to which the firms belong:
$$
\begin{array}{cccl}
N_i & \sim & \frac{350}{500}Beta(2,5) + \frac{100}{500}, &\quad\quad i\in T_1,\\
&&&\\
N_i &\sim & \frac{70}{500}Beta(2,5) + \frac{20}{500}, &\quad\quad i\in T_2,\\
&&&\\
N_i & \sim  & \frac{1.75}{500}Beta(2,5)+\frac{0.5}{500}, &\quad\quad i\in T_3.
\end{array}
$$
The firms in $T_1$, $T_2$, and $T_3$ are thus interpreted as large, medium-sized, and small, respectively. 
\item 
We assume that the acceptance criterion is formulated in terms of a shifted entropic risk measure, a special case of a utility-based risk measure (and, at the same time, an OCE-risk measure). More specifically, we choose the acceptance set
$$\acal = \{M \in L^0(\Omega; \bbr) \; | \; E[\exp (-M)] \leq \exp(-0.9) \}.$$
\item As in the previous example, we consider only capital allocations in $\bbr_+^l$.
\end{enumerate}

We investigate the resulting systemic risk measurements along two dimensions.   First we set the fraction of the initial endowment in the liquid asset to $\alpha = 60\%$ and compare the capital allocations for the large and mid-sized firms when varying the capital requirement for the small firms.  Figure~\ref{fig:3tier-cifuentes-3d}  depicts the contour map of the capital requirements for each group of firms.  This graph shows that in the chosen network model the capital requirements for the large and mid-sized firms are insensitive to the capital holdings of the small firms: all the contour lines are stacked on top of each other.  This observation provides a first indication on the design of regulatory rules for firms of different size which is consistent with current practice, as e.g.\ documented in \ci{dodd-frank}.

Second, due to the insensitivity to the capital requirements for small firms, we set the capital levels  of the small firms equal to zero in the next case study. We vary the portfolio composition by considering different values of $\alpha$, the percentage of the endowment in liquid assets.  We will choose $\alpha \in \{0\%,20\%,40\%,60\%,80\%,100\%\}.$  The minimal capital allocations for the large and mid-sized firms are plotted in Figure~\ref{fig:3tier-cifuentes}.  

The key finding is that when the fraction of the endowment invested in the liquid asset grows, the capital requirement for each firm shrinks as evidenced by the lower minimal capital allocations.  This is as expected, since the liquid assets have a constant value whereas the illiquid assets have a value that decreases as a function of the number of shares  sold.  In fact, due to the decreasing price of the illiquid asset as more shares are liquidated, we observe that the spacing between each scenario increases as the proportion invested in the liquid asset shrinks.  This is due to fact that the price impact becomes larger, the more illiquid assets need to be liquidated.

\end{example}

\section{Conclusion}
We proposed a novel systemic risk measure. On the basis of a model of the relevant stochastic outcomes  (CVM)  and an acceptability criterion, the systemic risk measures were defined as the set of capital allocations that achieve an acceptable outcome if added to the current capital. The suggested systemic risk measures can explicitly be computed. An algorithm was presented in the paper. Moreover, we introduced a new concept, called  EARs,  that admits a simple application of systemic risk measures to capital regulation in practice.  EARs  were characterized as minimizers of the cost of capital on a global level. We demonstrated that important recent contributions in the literature on systemic risk measures can be embedded into our framework. Finally, we illustrated in numerical case studies that the suggested risk measures are a powerful instrument for analyzing systemic risk.

\bibliography{bibtex2}
\bibliographystyle{jmr}
\vfill
\pagebreak

\appendix

\section{Proofs}\label{proofs}

\begin{proof}[Proof of Lemma~\ref{Lem:ci-m}]  These properties are a direct consequence of properties of $Y$ and $\acal$.
(i) follows immediately. The proof of (ii) uses property (ii) of the scalar acceptance set $\acal$.
Property (iii) follows from the fact that $Y$ is continuous and $\acal$ is closed.
\end{proof}

\begin{proof}[Proof of Lemma~\ref{Lem:conv-set}]
Let $m_1,m_2\in R(Y;k)$ for some $Y\in\ycal$ and $k\in\bbr^l$, that is, $Y_{k+m_1},Y_{k+m_2}\in\mathcal A$. Thus, due to the convexity of $\mathcal A$, the concavity of $k \mapsto Y_k$ and property (ii) of any acceptance set $\mathcal A$, one obtains for $\ga\in[0,1]$ that
$
\mathcal A\ni \ga Y_{k+m_1} + (1-\ga)Y_{k+m_2}\leq Y_{k+\ga m_1 + (1-\ga)m_2}$, implying that  $Y_{k+\ga m_1 + (1-\ga)m_2} \in \mathcal A.
$
Hence, $\ga m_1 + (1-\ga)m_2\in R(Y;k)$.
\end{proof}

\begin{proof}[Proof of Lemma~\ref{Propo:rm}]
This is a direct consequence of Lemma \ref{Lem:ci-m}(ii), the convexity of the scalar acceptance set $\acal$, and property \eqref{mixture} of the random fields $Y^i$ and $Y^j$.
\end{proof}

\begin{proof}[Proof of Proposition \ref{Propo:ortho-conv}]
Let $\bar{k}^*$ be any single-valued  EAR  corresponding to the systemic risk measure $R$. In particular, $\bar{k}^*(Y;k) \in \bbr^l$ for $Y \in \ycal$ and $k \in \bbr^l$.  By Lemma \ref{Propo:rm} and property \eqref{mixture} of the random fields $Y^i$ and $Y^j$, we have that $\bar{k}^*(Y^i;0) \vee \bar{k}^*(Y^j;0) \in R(D_\alpha(i,j);0)$.  

We will now construct a single-valued  EAR  $k^*$ from $\bar{k}^*$ which is $D$-quasi-convex at $(i,j)$.  Choosing $k^*(D_\alpha(i,j);0) \in [\bar{k}^*(Y^i;0) \vee k^*(Y^j;0) - \bbr^l_+] \cap \Min R(D_\alpha(i,j);0) \neq \emptyset$,   we define the  EAR  $k^*$ as follows:
\[k^*(Y;k) := \begin{cases} k^*(D_\alpha(i,j);0) - (k+m), & \text{if } {Y_\bullet = D_\alpha(i,j)_{m+\bullet} \text{ for }m\in\bbr^l  ,}\\ \bar{k}^*(Y;k), & \text{otherwise,}\end{cases} \]
for $ k \in \bbr^l$, $ Y \in \ycal$.
\end{proof}

\begin{proof}[Proof of Lemma \ref{thm:scalarization}]
First, observe that $\hat k(Y;k)$, as defined in \eqref{argminlemma}, is an  EAR;  indeed,  Definition~\ref{orthant}(i) is satisfied by Proposition~\ref{Propo:min} below, (ii) is true since $R$ has convex images, and (iii) is satisfied because $w: \ycal \to \bbr^l_{++}$ is independent of $k\in \bbr^l$.

Now, let $k^*(Y;k)$ be an arbitrary  EAR.  By Proposition~\ref{Propo:min}, any minimal element $k^* \in \Min R(Y;k)$ is a solution of an optimization problem that minimizes $\sum_{i = 1}^l w_i m_i $ over $ \; m \in R(Y;k)$ for some $w \in \bbr^l_{++} \cap {\rm recc}R(Y;k)^+$. 
Because of cash-invariance of systemic risk measures, we know that
$${\rm recc}R(Y;k) = {\rm recc}R(Y;0) \quad (\forall k \in \bbr^l, \; \forall Y \in \ycal).$$
Since $k^*(Y;k)$ is convex-valued, there exists $w(Y;k) \in \bbr^l_{++} \cap {\rm recc}R(Y;0)^+ $ such that
\begin{equation}
\label{eqproof}
k^*(Y;k) \subseteq \argmin\{w(Y;k)^{\T}m \; | \; m \in R(Y;k)\}.
\end{equation}
Finally, we will show that these $w(Y;k)$ can be chosen independently of $k \in \bbr^l$, i.e.\ $w(Y) = w(Y;k)$ for every $k \in \bbr^l$.
Letting $Y \in \ycal$ and $k^1,k^2 \in \bbr^l$, 
cash-invariance of $R(Y;\cdot)$ in the second component yields
\begin{align*}
k^*(Y;k^1) &\subseteq \argmin\{w(Y;k^1)^{\T}m \; | \; m \in R(Y;k^1)\}\\
&= \argmin\{w(Y;k^1)^{\T}m \; | \; m \in R(Y;k^2) + (k^2 - k^1)\}\\
&= \argmin\{w(Y;k^1)^{\T}m \; | \; m \in R(Y;k^2)\} + (k^2 - k^1).
\end{align*}
By cash-invariance of $k^*$, \eqref{eqproof} is satisfied if one picks $w(Y) := w(Y;k^1)=w(Y;k^2)$.
\end{proof}

\begin{proposition}\label{Propo:min}
Let $R(Y;k)$ be non-empty, closed and convex, and assume that  ${\rm recc}R(Y;k) \cap -\bbr^l_+ = \{0\}$. Then $k^* \in \Min R(Y;k)$, if and only if it is a minimizer of $w^\T m $ over $m \in R(Y;k)$ for some $w \in \bbr^l_{++} \cap {\rm recc}R(Y;k)^+$.
\end{proposition}

\begin{proof}
First, it is evident that if $k^* \in \bbr^l$ is a minimizer of $w^\T m $ over $ m \in R(Y;k)$ for some $w \in \bbr^l_{++}$, then $k^* \in \Min R(Y;k)$.

Conversely, assume that $k^* \in \Min R(Y;k)$.  By Lemma 2.2(ii) in \ci{hamel04lvo}, it follows that there exists a $w \in \bbr^l \backslash \{0\}$ such that
\[w^\T(m - k^* + \epsilon) > 0 > w^\T(-r) \quad (\forall m \in R(Y;k),\; \forall \epsilon \in \bbr^l_{++}, \; \forall r \in \bbr^l_+ \backslash \{0\}) , \]
since $(R(Y;k) - k^* + \epsilon) \cap -\bbr^l_+ = \emptyset$ and ${\rm recc}[R(Y;k) - k^* + \epsilon] \cap -\bbr^l_+ = {\rm recc}R(Y;k) \cap -\bbr^l_+ = \{0\}$.  Since $w^\T r > 0$ for all $r \in \bbr^l_+ \backslash \{0\}$, $w$ is an element of the quasi-interior of $\bbr^l_+$, i.e., $w \in \bbr^l_{++}$.  We also obtain $w^\T m > w^\T k^* - \delta$ for every $m \in R(Y;k)$ and $\delta \in \bbr_{++}$, because $\{w^\T \epsilon \; | \; \epsilon \in \bbr^l_{++}\} = \bbr_{++}$.  This implies that $w^\T m \geq w^\T k^*>-\infty$ for every $m \in R(Y;k)$, i.e., $k^*$ is the minimizer of $w^\T m$ over  $m \in R(Y;k)$, and $w$ must be an element of ${\rm recc}R(Y;k)^+$.
\end{proof}

\begin{proof}[Proof of Lemma~\ref{lemmae}] 
First, we observe that the function $\phi: [0,\bar p] \times [0,f(0)] \times \bbr^n_+ \times \bbr^n_+ \to [0,\bar p] \times [0,f(0)]$, defined by
\begin{align*}
\phi_i(\hat{p},\hat{\pi},x,s) &= \bar p_i \wedge \left( x_i + \hat{\pi}s_i + \sum_{j = 1}^n a_{ji} \hat{p}_j \right), \quad\quad i = 1,\dots,n, \\
\phi_{n+1}(\hat{p},\hat{\pi},x,s) &= f\left( \sum_{i = 1}^n \frac{(\bar{p} - x_i - \sum_{j = 1}^n a_{ji} \hat{p_j})^+}{\hat{\pi}} \wedge s_i \right),
\end{align*}
is continuous, and $\phi^s := \phi(\cdot,\cdot,\cdot,s)$ is non-decreasing.

\begin{enumerate}
\item Because $\phi^s$ is non-decreasing (for any $s \in \bbr^n_+$), $p$ and $\pi$ are non-decreasing in their first argument by Theorem 3 in \ci{MR94}.  Therefore, the  function $e: \bbr^n_+ \times \bbr^n_+ \to \bbr^{n+1}_+$, and in particular its component $e_0$, is non-decreasing in its first argument.

\item Since $\phi$ is a continuous function and $(p,\pi)$ is the unique fixed point of $\phi$ under our assumptions (see \ci{AFM15unique}, \ci{fei15}), it follows by Proposition~\ref{prop:cont-graph} below that 
$$\operatorname{graph}(p,\pi) = \{(x,s,\hat{p},\hat{\pi}) \in \bbr^n_+ \times \bbr^n_+ \times [0,\bar p] \times [0,f(0)]\; | \; \phi(\hat{p},\hat{\pi},x,s) = (\hat{p},\hat{\pi})\}$$
is closed in the product topology.  Let $\Psi: \bbr^n_+ \times \bbr^n_+ \times [0,\bar p] \times
[0,f(0)] \to \bbr^n_+ \times \bbr^n_+$ be the projection $\Psi(x,s,\hat{p},\hat{\pi}) = (x,s)$; this is a closed mapping
from the product topology (Proposition \ref{prop:proj-closed} below).  In order to show that $(p,\pi)$ is continuous, take $B \subseteq [0,\bar p] \times [0,f(0)]$ closed. Then
\begin{align*}
(p,\pi)^{-1}[B] &= \{(x,s) \in \bbr^n_+ \times \bbr^n_+\; | \; (p(x;s),\pi(x;s)) \in B\} 
= \Psi(\operatorname{graph}(p,\pi) \cap [\bbr^n_+ \times \bbr^n_+ \times B])
\end{align*}
is closed.  Hence, $p$ and $\pi$ are continuous functions.  Thus, the  function $e$, and in particular its component $e_0$, is continuous.

\item The function $e_0$ is bounded from below on $\bbr^n_+ \times \bbr^n_+$, since $p(x;s) \geq 0$ for every $x,s \in \bbr^n_+$. It is bounded from above, because $p(x;s) \leq \bar{p}$ by definition.

\item For this part, take $s = 0 \in \bbr^n$.  In this situation, we can ignore both $\pi$ from \eqref{price} in the
 function $e$ and the last component $\phi_{n+1}$ of the function $\phi$.

Consider the function $\phi:[0,\bar p]\times\bbr^n_+\rightarrow[0,\bar p]$ in \eqref{clearing}, defined by $\phi(\hat{p},x):= \bar p   \wedge   (A^{\T} \hat{p} +
x)$, and its fixed point $p(x)\in\bbr^n_+$. 
The function $\phi$ is non-decreasing in $\hat{p}$ and $x$, and jointly concave in $\hat{p}$ and $x$, as it is a
composition of a non-decreasing affine map ($(\hat{p},x) \mapsto A^{\T} \hat{p} + x$) and a non-decreasing concave map ($q \mapsto q \wedge \bar p$). As shown above, $p$ is non-decreasing in $x$. To show that $p$ is concave, define a sequence of functions $\{p_n(x)_{n=0}^\infty\}$ inductively as follows:
\[
p_n(x)=\phi(p_{n-1}(x),x),\quad p_0(x)=0\in\bbr^n.
\]
For each fixed $x\in\bbr^n_{+}$, the term $p_n(x)$ is the $n^{th}$ iteration of the map $\hat{p}\mapsto \phi(\hat{p},x)$ starting at $0$. Standard results on
the convergence of iterates of non-decreasing operators show that $p_n(x)\uparrow p(x)$, for all $x$. By induction, since  $\phi$ is
non-decreasing and jointly concave in $p$ and $x$, $p_n$ is concave for all $n\geq 0$. Then, $p$ is concave as the pointwise limit of concave
functions. The function $e: \bbr^n_{+}\rightarrow \bbr^n$ defined by $e(x):=A^{\T} p(x) + x-\bar p$ is the sum of a
non-decreasing affine function ($x \mapsto x-\bar p$) and a composition of a non-decreasing concave function ($x \mapsto p(x)$) and a
non-decreasing affine function ($q \mapsto A^{\T} q$). Thus, it is concave and non-decreasing. 
\end{enumerate}
\end{proof}

\begin{proposition}\label{prop:cont-graph}
Let $\bbx$ and $\bby$ be two Hausdorff topological linear spaces. Let $H: \bbx \times \bby \to \bby$ be a continuous function with unique fixed point $h(x)\in \bby$ for every $x \in \bbx$, i.e.\ $H(x, h(x)) = h(x)$ and for each $x\in \bbx$ the fixed point $h(x)\in \bby$ is unique.
Then, $\operatorname{graph}h$ is
closed in the product topology.
\end{proposition}
\begin{proof}
Let $(x_j,y_j)_{j \in J} \to (x,y)$ be a converging net such that $h(x_j) = y_j$ for every $j \in J$.  By continuity of $H$, it follows that $H(x,y) - y = \lim_j [H(x_j,y_j) - y_j] = 0$, i.e.\ $h(x) = y$.
\end{proof}

\begin{proposition}\label{prop:proj-closed}
Let $\Psi: \bbx \times \bby \to \bbx$ be the projection $\Psi(x,y) = x$ for every $x \in \bbx$ and $y \in \bby$ for a topological space $\bbx$ and a
compact space $\bby$, then $\Psi$ is a closed mapping from the product topology.
\end{proposition}
\begin{proof}
Let $A \subseteq \bbx \times \bby$ closed in the product topology and let $x \not\in \Psi[A]$, i.e.\ $\{x\} \times \bby\subseteq A^c$. By
the tube lemma, there exists an open neighborhood $N_x$ of $x$ so that $N_x \times \bby \subseteq A^c$.  This implies $N_x \subseteq
\Psi[A]^c$.  In particular, this is true for every $x \not\in \Psi[A]$, so $\Psi[A]^c$ is open, and thus the result is proven.
\end{proof}

\end{document}